\documentclass[10pt]{article}
\usepackage{amsmath,amsthm,amssymb}
\usepackage[utf8]{inputenc}
\usepackage{graphicx}
\usepackage{color}
\theoremstyle{plain}
\newtheorem{Theorem}{Theorem}
\newtheorem{Lemma}{Lemma}
\newtheorem{Proposition}{Proposition}
\newtheorem{Corollary}{Corollary}
\newtheorem{Remark}{Remark}
\newtheorem{Condition}{Condition}

\numberwithin{equation}{section}

\begin{document}
\author
{
    S. A. Simonov
    \thanks
    {
        St. Petersburg Department of V. A. Steklov Institute of Mathematics of the Russian Academy of Sciences, 27 Fontanka, St. Petersburg, 191023 Russia;
        St. Petersburg State University, 7--9 Universitetskaya nab., St. Petersburg, 199034 Russia,
        sergey.a.simonov@gmail.com.
        The work was supported by the grant RFBR 18-31-00185mol\_a.
    }
}
\title{The wave model of the Sturm--Liouville operator on an interval}
\maketitle
\begin{abstract}
    In the paper we construct the wave functional model of a symmetric
    restriction of the regular Sturm-Liouville operator on an interval. The
    model is based upon the notion of the wave spectrum and is constructed
    according to an abstract scheme which was proposed earlier. The result of
    the construction is a differential operator of the second order on an interval,
    which differs from the original operator only by a simple transformation.
\end{abstract}



\section*{Introduction}\label{sec Introduction}
In the work \cite{Belishev-2013} the notion of the wave spectrum of a symmetric semibounded from below operator was introduced. The wave spectrum is constructed as a topological space determined by the operator. In the same work the wave spectrum was studied for the Laplace operator on a compact manifold and it was established that in the general situation one can introduce a metric on the wave spectrum so that it becomes isometric to the original manifold. In \cite{Belishev-Simonov-2017} a scheme of construction of a functional model of such an operator was proposed, which is called {\it the wave model} and is based on the notion of the wave spectrum. The space of functions on the wave spectrum is taken as the model space. The graph of the model operator is recovered using the method of boundary control, on which the construction of the wave spectrum also relies. This scheme was realized in \cite{Belishev-Simonov-2017} for the positive-definite Schr\"odinger operator on the half-line in the limit point case. To be precise, the wave model was constructed for a symmetric restriction of such an operator with defect indices $(1,1)$.

In the present paper we construct the wave model of a symmetric positive-definite operator with defect indices $(2,2)$, namely, of the symmetric restriction of the regular Sturm--Liouville operator defined by the differential expression $-\frac{d^2}{dx^2}+q(x)$ on the interval $(0,l)$ with the boundary conditions $u(0)=u'(0)=u(l)=u'(l)=0$. The potential $q$ is supposed to be smooth and such that the operator is positive-definite. In course of construction we also refine and develop the abstract scheme of the wave model.

The paper consists of two parts. In the first, abstract, part we give the definition of the wave spectrum and describe the scheme of the wave model construction. Trying to keep certain level of generality, we formulate a number of conditions, which a symmetric operator should satisfy and under which the wave model is constructed. Conditions are formulated in rather abstract terms, thus one can check them only constructing the model of some particular operator. The second part of the paper is devoted to realization of this abstract scheme for the Sturm--Liouville operator on an interval. We explicitly describe objects that were defined in the first part and directly check all the conditions.

An important feature of the wave functional model is that it turns out to be almost identical to the original operator. This happens for the example considered earlier \cite{Belishev-Simonov-2017}
and in our case. The inverse problem data (spectral, dynamical) in some cases allows to construct certain {\it ``auxiliary model''}, i.e., a model space and an operator acting on it which is unitarily equivalent to the original operator. In this sense one can distinguish objects that are available to the {\it ``outer observer''} (those which can be obtained knowing the inverse problem data) and available to the {\it ``inner observer''} (those that can be obtained knowing the original, the solution of the inverse problem). Knowing the auxiliary model, the ``outer observer'' can construct its wave model, from which it is easy to recover the original. In our examples the wave model is a differential operator. From the coefficients of this operator one can explicitly obtain the potential of the original operator. In the case of the regular Sturm--Liouville operator the potential can be recovered up to reflection from the middle point of the interval $(0,l)$.

The results of this work were announced in \cite{Simonov-2017}, where a brief description of this construction was given.

\section{The abstract scheme}
\subsection{The operator $L_0$}\label{ssec Operator $L_0$}
Consider a closed symmetric linear operator $L_0$ in a separable Hilbert space $\cal H$, and let this operator be positive-definite: there exists $\varkappa>0$ such that $(L_0 u,u)\geqslant \varkappa \|u\|^2$ for every $u\in{\rm Dom\,}L_0$. Denote by $L$ the Friedrichs self-adjoint extension of the operator $L_0$, \cite{Birman-Solomyak-1980}. For every $u\in{\rm Dom\,}L$ one has $(Lu,u)\geqslant\varkappa \|u\|^2$,  hence   the bounded inverse operator $L^{-1}$ exists.

\subsection{The Green's system}\label{ssec System Gr_L_0}
Let $A$ be an operator in $\cal H$, $\cal B$ a Hilbert space, $\Gamma_1$ and $\Gamma_2$ be linear operators acting from $\cal H$ to $\cal B$. Let the following conditions hold:
$$
    \overline{{\rm Dom\,}A}={\cal H},
    \quad{\rm Dom\,}A\subseteq{\rm Dom\,}\Gamma_1\cap{\rm Dom\,}\Gamma_2,
    \quad\overline{{\rm Ran\,}\Gamma_1+{\rm Ran\,}\Gamma_2} = \cal B.
$$
The collection ${\mathfrak G}=\{{\cal H}, {\cal B}; A, \Gamma_1,\Gamma_2\}$ is called the {\it Green's system}, if the equality
\begin{equation}\label{Green Formula}
    (Au,v)_{\cal H}-(u,Av)_{\cal H}=(\Gamma_1u, \Gamma_2 v)_{\cal B}-(\Gamma_2u, \Gamma_1 v)_{\cal B}
\end{equation}
(the {\it Green's formula}) holds for every $u, v \in {\rm Dom\,} A$, \cite{Kochubei-1975,Derkach-Malamud-1995,Ryzhov-2007}. The space $\cal H$ is called the {\it inner space}, $\cal B$ the {\it space of boundary values}, $A$ the {\it basic operator}, $\Gamma_1,\Gamma_2$ the {\it boundary operators}.

There is a class of Green's systems which canonically corresponds to the class of operators $L_0$ that we consider. Denote
$$
    {\cal K}\,:=\,{\rm Ker\,} L_0^*,
$$
let $P_K$ be the orthogonal projection on the subspace ${\cal K}$ of $\cal H$, $\mathbb O$ be the zero operator in $\cal H$, $\mathbb I$ be the identity operator. Let
\begin{equation}\label{Ga1Ga2}
    \Gamma_1\,:=\,L^{-1}L_0^*-{\mathbb I}\,, \quad \Gamma_2\,:=\,P_K L_0^*.
\end{equation}
Then the collection ${\mathfrak G}_{L_0}:=\{{\cal H}, {\cal K};\,L_0^*, \Gamma_1, \Gamma_2\}$ forms a Green's system \cite{Belishev-Demchenko-2012}. Such a system is related to the {\it Vishik's decomposition} for the operator $L_0$, which has the form
\begin{equation}\label{Vishik general}
    {\rm Dom\,}L_0^*={\rm Dom\,}L_0 \overset{.}+L^{-1}{\cal K}\overset{.}+{\cal K}
\end{equation}
($\overset{.}+$ denotes the direct sum of linear sets). The boundary operators can be written in terms of this decomposition as follows \cite{Vishik-1952}: if $u\in {\rm Dom}L_0^*$ is represented in the form
\begin{equation}\label{Vishik concrete}
    u=u_0+L^{-1}g_u + h_u,
\end{equation}
where $u_0\in{\rm Dom}\, L_0,g_u, h_u \in {\mathcal K}$, then
\begin{equation}\label{**}
    \Gamma_1 u = -h_u, \quad \Gamma_2 u = g_u.
\end{equation}

\subsection{The system with boundary control}\label{ssec System alpha}
Consider the following problem, which corresponds to the Green's system ${\mathfrak G}_{L_0}$:
\begin{align}
    \label{alpha1} & u_{tt}+L_0^*u = 0,  && t>0,
    \\
    \label{alpha2} & u|_{t=0}=u_t|_{t=0}=0, &&
    \\
    \label{alpha3}& \Gamma_1 u = h\,, && t\geqslant 0,
\end{align}
where $h=h(t)$, a $\cal K$-valued function, is called the {\it boundary control}, and the ${\cal H}$-valued function $u(t)=u^h(t)$ is unknown. In the control theory $u^h(\cdot)$ is called the {\it trajectory}, $u^h(t)$ the {\it state} of the system at the moment $t$; we will call $u^h$ the {\it wave}. Denote the system (\ref{alpha1})--(\ref{alpha3}) by ${\alpha_{L_0}}$.

The problem (\ref{alpha1})--(\ref{alpha3}) has a solution, \cite{Belishev-Demchenko-2012}, if the control $h$ belongs to the class
\begin{equation}\label{linal cal M}
    {\cal M}:=\{h \in C^{\infty}\left([0,\infty); {\cal K}\right):{\rm supp\,}h \subset (0,\infty)\}.
\end{equation}
This solution can be written in the form
\begin{equation}
    \label{weak solution u^f}
    u^h(t)=-h(t)+\int_0^t L^{-\frac12}\,\sin\left[(t-s)L^{\frac12}\right]\,h_{tt}(s)\,ds\,,\qquad t \geqslant 0,
\end{equation}
it belongs to $C^{\infty}\left([0,\infty); {\cal H}\right)$ and vanishes near zero. We will call such $u^h$ {\it classical solutions} or {\it smooth waves}.

The set of states of the system $\alpha_{L_0}$
\begin{equation}\label{U^t}
    {\cal U}^t_{L_0}:=\{u^h(t),h\in{\cal M}\}\subseteq {\rm Dom\,}L_0^*
\end{equation}
is called the {\it reachable set} at the time $t\geqslant0$. It is easy to see that ${\cal U}_{L_0}^t$ grows with $t$. The set
\begin{equation}\label{U}
    {\cal U}_{L_0}:=\bigcup_{t>0}{\cal U}^t_{L_0}
\end{equation}
is called the {\it total reachable set} of the system $\alpha_{L_0}$, and its orthogonal complement
$$
    {\cal D}_{L_0}\,:=\,{\cal H} \ominus \overline{\cal U}_{L_0}
$$
is called the {\it defect subspace} of the system $\alpha_{L_0}$. Linear sets ${\cal U}^t_{L_0}$ and ${\cal U}_{L_0}$ are invariant under $L_0^*$: let $T\geqslant0$ and $u=u^h(T)\in {\cal U}_{L_0}^T$, then
\begin{align*}
    &L_0^*u^h(T) \overset{(\ref{alpha1})}=-u^h_{tt}(T)\overset{(\ref{weak solution u^f})}=-u^{h_{tt}}(T)\in{\cal U}_{L_0}^T,
    \\
    &u^h(T) \overset{(\ref{alpha2})}=J^2[u^h_{tt}](T)\overset{(\ref{weak solution u^f})}=u^{J^2h}_{tt}(T)\overset{(\ref{alpha1})}=-L_0^*u^{J^2h}(T)\in L_0^*{\cal U}_{L_0}^T\,,
\end{align*}
where $J: u\mapsto\int_0^tu(s)ds$. Therefore $L_0^*{\cal U}_{L_0}= {\cal U}_{L_0}$.

The system $\alpha_{L_0}$ is called {\it controllable}, if $\overline{{\cal U}}_{L_0}={\cal H}$. The following fact is known \cite{Belishev-Demchenko-2012}.

\begin{Proposition}
    Controllability of the system ${\alpha_{L_0}}$ is equivalent to the fact that the operator $L_0$ is completely non-selfadjoint.
\end{Proposition}

The restriction of the operator $L_0^*$ to the linear set of smooth waves ${\cal U}_{L_0}\subseteq{\rm Dom}\,L_0^*$ is not necessarily a closed operator. Its closure $\overline{L_0^*|_{{\cal U}_{L_0}}}\subseteq L_0^*$ is called the {\it wave part} of the operator $L_0^*$. If the operator $L_0$ is completely non-selfadjoint, the question arises whether the operator $L_0^*$ coincides with its wave part. This happens for the examples that we know, however, we do not have a proof of the general fact.

\subsection{The wave spectrum}\label{sec Wave spectrum}
The functional model of the operator $L_0$ that we construct is based on the wave spectrum of the operator. For its definition we use notions of lattice theory.

{\it Lattice} is a partially ordered set every two elements $p,q$ of which have the least upper bound ${\rm sup}\{p,q\}=p \vee q$ (the least element of the set of all upper bounds) and the greatest lower bound ${\rm inf}\{p,q\}=p\wedge q$ (the greatest element of the set of all lower bounds). A lattice is called {\it complete}, if every its subset has the least upper and the greatest lower bounds. In a complete lattice there always exist the least and the greatest elements.

Let $\cal P$ and $\cal Q$ be partially ordered sets, $i$ be a map from $\cal P$ to $\cal Q$. The map $i$ is called {\it isotonic}, if $p_1\leqslant p_2$ in $P$ implies $i(p_1)\leqslant i(p_2)$ in $Q$, \cite{Birkhoff-1984}. We call {\it isotony} a family of maps $\{i^t\}_{t\geqslant0}$ from $P$ to $Q$, if $p_1\leqslant p_2$ and $t_1\leqslant t_2$ implies $i^{t_1}(p_1)\leqslant i^{t_2}(p_2)$.

Let $P=Q=\mathfrak L$ be a complete lattice and $O_{\mathfrak L}$ be its least element. Then an isotony $I^t$ is called an {\em isotony of the lattice $\mathfrak L$}, if $I^0$ is the identity map in $\mathfrak L$ and $I^t(O_{\mathfrak L})=O_{\mathfrak L}$ for every $t\geqslant 0$.

Let a partially ordered set $\cal P$ contain the least element $O_{\cal P}$. An element $p\neq O_{\cal P}$ is called an {\it atom} of $\cal P$, if there is no element $p\,'\in\cal P$ such that $O_{\cal P}<p\,'<p$.

Let $\mathfrak L$ be a complete lattice, $O_{\mathfrak L}$ be its least element, $E_{\mathfrak L}$ be its greatest element. If for every $p\in\mathfrak L$ there exists an element $p\,'\in\mathfrak L$ such that $p\vee p\,'=E_{\mathfrak L}$, $p\wedge p\,' = O_{\mathfrak L}$ (the {\it complement}), then $\mathfrak L$ is called a {\it lattice with complements}.

\subsubsection{The lattice of subspaces}
We will work with lattices and isotonies of a special kind.

The set ${{\mathfrak L}(\cal H)}$ of all subspaces of a Hilbert space $\cal H$ with the partial order $\subseteq$ forms a complete lattice with complements: it is easy to check that $\mathcal G_1\vee\mathcal G_2=\overline{\mathcal G_1+\mathcal G_2}$ and $\mathcal G_1\wedge\mathcal G_2={\mathcal G_1}\cap{\mathcal G_2}$ for $\mathcal G_1,\mathcal G_2\in\mathfrak L(\mathcal H)$, $\{0\}$ is the least element, $\cal H$ is the greatest element, $\mathcal G^{\bot}$ is the complement for $\mathcal G\in\mathfrak L(\cal H)$.

Let us call $\mathfrak L\subset\mathfrak L(\cal H)$ a {\it sublattice} of the lattice $\mathfrak L(\mathcal H)$ {\it with complements}, if $\mathfrak L$ contains $\{0\}$, $\cal H$, $\mathcal G_1\vee\mathcal G_2$, $\mathcal G_1\wedge\mathcal G_2$ for every $\mathcal G_1,\mathcal G_2\in\mathfrak L$ and $\mathcal G^{\bot}$ for every $\mathcal G\in\mathfrak L$. For every subset $\mathfrak M\subset\mathfrak L(\mathcal H)$ there exists the minimal sublattice with complements $\mathfrak L_{\mathfrak M}$ in $\mathfrak L(\mathcal H)$, which contains $\mathfrak M$. If $I^t$ is an isotony of the lattice $\mathfrak L(\mathcal H)$, then there also exists the minimal lattice with complements $\mathfrak L_{\mathfrak M}^{I}$ in $\mathfrak L(\mathcal H)$, which contains $\mathfrak M$ and is invariant under $I$: for every $\mathcal G\in\mathfrak L_{\mathfrak M}^{I}$ and $t\geqslant 0$ one has $I^t(\mathcal G)\in\mathfrak L_{\mathfrak M}^{I}$, \cite{Belishev-Simonov-2017}.

One can naturally define a topology on the lattice of subspaces $\mathfrak L(\mathcal H)$. A sequence $\{\mathcal G_n\}_{n\in\mathbb N}$ from $\mathfrak L(\mathcal H)$ converges to $\mathcal G\in\mathfrak L(\mathcal H)$ as $n\rightarrow\infty$, if the corresponding projections converge in the strong sense: $P_{\mathcal G_n}\overset{s}\to P_{\mathcal G}$. Note that the strong operator topology, restricted to orthogonal projections, satisfies the first axiom of countability and can be described in terms of converging sequences, \cite{Kim-2006}.

Let $\mathfrak F(\mathcal H)$ denote the set of functions from $[0,\infty)$ to ${{\mathfrak L}(\mathcal H)}$ with the pointwise partial order: $f_1\leqslant f_2$, if $f_1(t)\leqslant f_2(t)$ in $\mathfrak L(\mathcal H)$ (i.e., $f_1(t)\subseteq f_2(t)$) for every $t\geqslant 0$. Then the lattice operations will also be pointwise:
\begin{align*}
    &(f_1\vee f_2)(t)=f_1(t)\vee f_2(t),
    \\
    &(f_1\wedge f_2)(t)=f_1(t)\wedge f_2 (t),
    \\
    &(f^\bot)(t)=(f(t))^\bot.
\end{align*}
Strong operator topology generates on $\mathfrak F(\mathcal H)$ the product topology (the topology of pointwise convergence), which does not satisfy the first axiom of countability and can be described in terms of converging nets instead of sequences. It turns out that the objects that we work with do not require a topology on $\mathfrak F(\mathcal H)$ and that it is possible to deal with the operation of sequential closure (topology corresponding to such an operation can be not unique). There exists a version of our construction of the wave model based on the product topology in $\mathfrak F(\mathcal H)$. For all the examples known to us, both versions eventually lead to the same construction (because the wave spectra coincide).

Let us denote by $I\mathfrak L(\mathcal H):=\{I^t(\mathcal G), \mathcal G\in\mathfrak L(\mathcal H) \}$ the set of isotonic $\mathfrak L(\mathcal H)$-valued functions, obtained by applying the isotony $I$ to the elements of the lattice $\mathfrak L(\mathcal H)$. We denote by $[I\mathfrak L(\mathcal H)]_{\rm seq}$ the sequential closure of this set in $\mathfrak F(\mathcal H)$.

\begin{Lemma}\label{Lem iso}
    Let $I$ be an isotony of the lattice $\mathfrak L(\mathcal H)$. Then the elements of $[I{\mathfrak L(\mathcal H)}]_{\rm seq}$ are isotonic functions.
\end{Lemma}

\begin{proof}
Let $f\in[I{\mathfrak L(\mathcal H)}]_{\rm seq}$. Then there exists a sequence $\{\mathcal G_n\}_{n\in\mathbb N}$ such that for every $t\geqslant 0$ one has
$$
    f(t)=\mathfrak L(\mathcal H)-\lim\limits_{n\rightarrow\infty}I^t(\mathcal G_n).
$$
Let $t_1\leqslant t_2$. Then
$$
    P_{f(t_1)}=s-\lim\limits_{n\rightarrow\infty}P_{I^{t_1}(\mathfrak G_n)},\quad P_{f(t_2)}=s-\lim\limits_{n\rightarrow\infty}P_{I^{t_2}(\mathfrak G_n)}.
$$
From the inclusion $I^{t_1}(\mathcal G_n)\subseteq I^{t_2}(\mathcal G_n)$ and the equality $P_{I^{t_2}(\mathcal G_n)}P_{I^{t_1}(\mathcal G_n)}=P_{I^{t_1}(\mathcal G_n)}$ we obtain for every $x$ that
\begin{multline*}
    P_{I^{t_2}(\mathcal G_n)}P_{I^{t_1}(\mathcal G_n)}x
    =P_{I^{t_2}(\mathcal G_n)}(P_{I^{t_1}(\mathcal G_n)}-P_{f(t_1)})x
    \\
    +(P_{I^{t_2}(\mathcal G_n)}-P_{f(t_2)})P_{f(t_1)}x
    +P_{f(t_2)}P_{f(t_1)}x.
\end{multline*}
Owing to convergence and to boundedness of the norms, $\|P_{I^{t_2}(\mathcal G_n)}\|=1$, we get $P_{I^{t_2}(\mathcal G_n)}P_{I^{t_1}(\mathcal G_n)}\overset{s}\to P_{f(t_2)}P_{f(t_1)}$ as $n\rightarrow\infty$    and $P_{f(t_2)}P_{f(t_1)}=P_{f(t_1)}$, which implies the inclusion $f(t_1)\subseteq f(t_2)$.
\end{proof}

Define the ``balls'' in the set $[I{\mathfrak L(\mathcal H)}]_{\rm seq}$
$$
    B_r(f)=\{g\in[I{\mathfrak L(\mathcal H)}]_{\rm seq}: \exists t>0: g(t)\neq0,g(t)\subset f(r)\}.
$$

\begin{Lemma}\label{Lem base}
Let $I$ be an isotony of the lattice $\mathfrak L(\mathcal H)$. Then the system of sets $\{B_r(f),\ f\in[I{\mathfrak L(\mathcal H)}]_{\rm seq},\ r>0\}$ is a base of some topology on $[I{\mathfrak L(\mathcal H)}]_{\rm seq}$.
\end{Lemma}

\begin{proof}
Let us check the condition for a family of sets to be a base of topology: let $f\in B_{r_1}(f_1)\cap B_{r_2}(f_2)$.
Prove that there exists a radius $r$ such that $B_{r}(f)\subseteq B_{r_1}(f_1)\cap B_{r_2}(f_2)$.
There exist $t_1$ and $t_2$ such that $f(t_1),f(t_2)\neq\{0\}$, $f(t_1)\subseteq f_1(r_1)$, $f(t_2)\subseteq f_2(r_2)$.
Since, by Lemma \ref{Lem iso}, $f$ is an isotonic function, $f(r)\subseteq f_1(r_1)\cap f_2(r_2)$ for $r:={\rm min}\{t_1,t_2\}$, while $f(r)\neq\{0\}$.
Then for every $g\in B_r(f)$ there exists $t_g>0$ such that $g(t_g)\neq\{0\}$ and $g(t_g)\subseteq f(r)\subseteq f_1(r_1)\cap f_2(r_2)$, so that $g(t_g)\subseteq f_1(r_1)$ and $g(t_g)\subseteq f_2(r_2)$.
This means that $g\in B_{r_1}(f_1)$ and $g\in B_{r_2}(f_2)$, and so $g\in B_{r_1}(f_1)\cap B_{r_2}(f_2)$. The lemma is proved.
\end{proof}

\begin{Remark}
If instead of $[I{\mathfrak L(\mathcal H)}]_{\rm seq}$ one considers $\overline{I\mathfrak L(\mathcal H)}$, the closure of the set of functions $I\mathfrak L(\mathcal H)$ in the topology of pointwise convergence on the lattice $\mathfrak F$, then analogs of Lemmas \ref{Lem iso} and \ref{Lem base} will hold. In the proof of Lemma \ref{Lem iso} in this case one should substitute the sequence $\{\mathcal G_n\}_{n\in\mathbb N}$ by the net $\{\mathcal G_{\alpha}\}$. The ball topology on $[I{\mathfrak L(\mathcal H)}]_{\rm seq}\subset\mathfrak F(\mathcal H)$ clearly differs from the topology of pointwise convergence.
\end{Remark}

\subsubsection{The wave isotony}
For every positive-definite self-adjoint operator $A$ one can define an isotony of the lattice $\mathfrak L(\mathcal H)$ in the following way. Consider the system
\begin{align}
    \label{dual 1} & v_{tt}+A v = g\,,  && t>0,
    \\
    \label{dual 2} & v|_{t=0}=v_t|_{t=0}=0\,,
\end{align}
where $g$ is an $\cal H$-valued function of time. If $g\in C^{\infty}\left([0,\infty); {\cal H}\right)$, then this problem has the unique solution $v=v^g(t)$ given by the Duhamel's formula \cite{Birman-Solomyak-1980}:
\begin{align}\label{v^g general}
    v^g(t)=\int_0^tA^{-\frac{1}{2}}\,\sin\left[(t-s)A^{\frac{1}{2}}\right]g(s)ds.
\end{align}
Let $\cal G \in {{\mathfrak L}(\cal H)}$. Consider the sets
\begin{equation}\label{V^t}
    \mathcal V^t_{A}(\mathcal G):=\left\{v^g(t),g\in C^{\infty}([0,t];{\mathcal G})\right\}
\end{equation}
and define the family of maps $\{I_{A}^t\}_{t\geqslant 0}$ as follows:
\begin{align*}\label{I^t}
    I_{A}^0&:={\rm id};
    \\
    I_{A}^t(\mathcal G)&:=\overline{{\mathcal V}^t_A(\mathcal G)},t>0.
\end{align*}

\begin{Proposition}[\cite{Belishev-2013}]\label{Prop 1}
The family $\{I_A^t\}_{t\geqslant 0}$ is an isotony of the lattice ${\mathfrak L}(\mathcal H)$.
\end{Proposition}

We call such an isotony $I_{A}^t$ the {\it wave isotony} of the lattice $\mathfrak L(\mathcal H)$ defined by the operator $A$.

\subsubsection{The wave spectrum}
Let us return to the original problem. The family of reachable sets of the system $\alpha_{L_0}$ defines the family of subspaces $\mathfrak M_{L_0}=\{\mathcal U_{L_0}^t,t\geqslant 0\}\subset\mathfrak L(\mathcal H)$, and the operator $L$ defines the wave isotony $I_L^t$. As we mentioned above, there exists the minimal sublattice $\mathfrak L_{L_0}$ in $\mathfrak L(\mathcal H)$ which contains the family   $\mathfrak M_{L_0}$ and is invariant under $I_L^t$. Denote $I_L\mathfrak L_{L_0}=\{I_L(\mathcal G), \mathcal G\in\mathfrak L_{L_0}\}$, also denote by $[I_L\mathfrak L_{L_0}]_{\rm seq}$ the (sequential) closure of this set in $\mathfrak F(\mathcal H)$.

The {\it wave spectrum} $\Omega_{L_0}$ of the operator $ L_0$ is the set of atoms of the partially ordered set $[I_L\mathfrak L_{L_0}]_{\rm seq}$,
$$
    \Omega_{L_0}:={\rm At}[I_L\mathfrak L_{L_0}]_{\rm seq}.
$$
The wave model of the operator $L_0$, which is a unitarily equivalent operator in the model space, requires for its construction some additional conditions on $L_0$. Examples that we considered earlier \cite{Belishev-2013,Belishev-Simonov-2017} suggest that the wave model can be constructed for some class of differential operators. In course of construction we formulate these additional general conditions on $L_0$ using notions that we gradually introduce.

\begin{Condition}\label{condition 1 nonempty}
The wave spectrum of the operator $L_0$ is not empty: $\Omega_{L_0}\neq\emptyset$.
\end{Condition}

The ball topology on $[I_L\mathfrak L(\mathcal H)]_{\rm seq}$ induces a topology on the wave spectrum. Under additional assumptions on $\Omega_{L_0}$ one can also define a metric (in the examples mentioned above the ``balls'' $B_r(f)$ turn out to be open balls in this metric). Each atom $\omega\in\Omega_{L_0}$, being a function from $[0,\infty)$ to $\mathfrak L(\mathcal H)$, defines a non-decreasing family of projections $P_{\omega(t)}$. If $P_{\omega(t)}\overset{s}\to I$ as $t\rightarrow+\infty$, then one can consider the self-adjoint and, generally speaking, unbounded operator
$$
    \tau_{\omega}=\int_{0}^{\infty}tdP_{\omega(t)},
$$
the {\it eikonal}. It can happen that even for unbounded $\tau$ the following holds.

\begin{Condition}\label{condition 2 eikonal}
    $P_{\omega(t)}\overset{s}\to I$ as $t\rightarrow+\infty$ for every $\omega\in\Omega_{L_0}$, and $\tau_{\omega_1}-\tau_{\omega_2}$ is a bounded operator in $\mathcal H$ for every $\omega_1,\omega_2\in\Omega_{L_0}$.
\end{Condition}

In such a case one can consider the function
$$
    \tau(\omega_1,\omega_2)=\|\tau_{\omega_1}-\tau_{\omega_2}\|
$$
as a distance in $\Omega_{L_0}$ (the properties of distance can be checked easily). For the wave spectrum one can also define the ``boundary'' $\partial\Omega_{L_0}$ as the set
$$
    \partial\Omega_{L_0}:=\{\omega\in\Omega_{L_0}\ :\ \forall t>0\ \omega(t)\subseteq\overline{\mathcal U_{L_0}^t}\}.
$$
In the case of the Laplace operator on a compact Riemannian manifold the ``boundary'' of the wave spectrum corresponds to the boundary of the manifold \cite{Belishev-2013}.

\subsection{The wave model}\label{sec Model}
Our goal is to construct the wave model so that this construction is applicable not only to the operator $L_0$, but also to its unitary copies. For this it is important to ensure that the wave model is constructed using the objects which are available to the ``outer observer''.

\subsubsection{The wave representation}
If Conditions 1 and 2 hold for the operator $L_0$, then its wave spectrum is a metric space with the distance $\tau$. The model space for the wave model should consist of functions on $\Omega_{L_0}$ which take values in some ``natural'' auxiliary spaces.

The first step in constructing the model space are spaces of germs on atoms. For every $\omega\in\Omega_{L_0}$ consider the following equivalence relation on $\mathcal H$: $u_1\underset{\omega}\sim u_2$, if there exists $t>0$ such that $P_{\omega(t)}u_1=P_{\omega(t)}u_2$. Corresponding equivalence classes $\tilde u(\omega)$ are called {\it germs}. Germs form the linear space which we denote by $\tilde{\mathcal H}_{\omega}$ and call the {\it stalk} above $\omega$. Consider the space of functions on the wave spectrum, which take values in stalks
$$
    \tilde{\mathcal H}:=\{\tilde u(\cdot),u\in\mathcal H\}.
$$
We need the operator $W:u\mapsto\tilde u\in\tilde{\mathcal H}$ to be bijective from $\mathcal H$ to $\tilde{\mathcal H}$, and for this the following condition is imposed, which we call {\it completeness of the system of atoms} of the wave spectrum.

\begin{Condition}\label{condition 3 complete system}
   For every nonzero $u \in \cal H$ there exists an atom $\omega\in\Omega_{L_0}$ such that $P_{\omega(\varepsilon)}u\neq0$ for every $\varepsilon>0$.
\end{Condition}

It is not convenient to work with this space, because stalks have infinite dimension. Besides that there is no Hilbert structure there. Thus we need additional conditions. Possibility to factorize further in germs is related to existence of gauge elements in $\mathcal H$. In order to define them, we need the following condition of {\it vanishing of atoms at zero}.

\begin{Condition}\label{condition 4 vanishing at zero}
  $\omega(t)\overset{\mathfrak L(\mathcal H)}\longrightarrow\{0\}$ as $t\rightarrow+0$ for every $\omega\in\Omega_{L_0}$.
\end{Condition}

By Lemma \ref{Lem iso} this is equivalent to the condition $\bigcap\limits_{t>0}\omega(t)=\{0\}$ for every atom. We call an element $e\in\mathcal H$ a {\it gauge} element of the operator $L_0$, if there exists a set of atoms $\Omega_{L_0}^{e}\subseteq\Omega_{L_0}$ such that its elements form a complete system in the sense of Condition \ref{condition 3 complete system} and that for every $u\in\mathcal U_{L_0}$ and $\omega\in\Omega_{L_0}^{e}$ the following limit exists:
$$
    \lim\limits_{t\to+0}\frac{\|P_{\omega(t)}u\|_{\mathcal H}}{\|P_{\omega(t)} e\|_{\mathcal H}}.
$$
As we see, the linear set of smooth waves starts playing an important role here.

\begin{Condition}\label{condition 5 gauge element}
  The operator $L_0$ has a gauge element.
\end{Condition}

Let $\omega\in\Omega_{L_0}^e$. For every $u,v\in{\mathcal U}_{L_0}$ the limit
$$
    \langle u,v\rangle_{\omega}:=\lim\limits_{t\to+0}\frac{(P_{\omega(t)}u,v)_{\mathcal H}}{(P_{\omega(t)} e,e)_{\mathcal H}}
$$
exists. It can be considered as a non-negative sesquilinear form on $\tilde{\mathcal U}_{L_0,\omega}:=\{\tilde u(\omega), u\in{\mathcal U}_{L_0}\}$, a linear set in the stalk above $\omega$. After factorization of $\tilde{\mathcal U}_{L_0,\omega}$ by the neutral subspace $\tilde{\mathcal U}^{\,0}_{L_0,\omega}$ of this form we obtain the linear space $\tilde{\mathcal U}_{L_0,\omega}/\tilde{\mathcal U}_{L_0,\omega}^{\,0}$. Denote its elements by $[u](\omega)$, $u\in\mathcal U_{L_0}$. This space has the inner product
$$
    \langle[u](\omega),[v](\omega)\rangle_{\mathcal U_{L_0,\omega}^{\rm w}}=\langle u,v\rangle_{\omega}.
$$
After completion in the corresponding norm we obtain the {\it space of values} $\mathcal U_{L_0,\omega}^{\rm w}$.

\begin{Condition}\label{condition 6 measure on spectrum}
    There exists a measure $\mu$ on $\Omega_{L_0}$ such that $\mu(\Omega_{L_0}\backslash\Omega^e_{L_0})=0$ and the equality
    \begin{equation}\label{isometry u to u omega}
        (u,v)_{\mathcal H}=\int_{\Omega_{L_0}}\langle[u](\omega),[v](\omega)\rangle_{\mathcal U_{L_0,\omega}^{\rm w}} d\mu(\omega)
    \end{equation}
    holds for every $u,v\in{\cal U}_{L_0}$.
\end{Condition}

The space
$$
    {\mathcal H}^{\rm w}:=\oplus\int_{\Omega_{L_0}}{\mathcal U}_{L_0,\omega}^{\rm w}\,d\mu(\omega)
$$
is called the {\it wave representation} of the space $\mathcal H$. For the operator $W_0^{\rm w}:u\mapsto[u](\cdot)$ which acts from $\mathcal U_{L_0}$ to ${\mathcal H}^{\rm w}$ one has $\|W_0^{\rm w}u\|_{\mathcal H}=\|[u]\|_{{\mathcal H}^{\rm w}}$ owing to \eqref{isometry u to u omega}, therefore the operator $W^{\rm w}=\overline{W_0^{\rm w}}$ is isometric.

\begin{Condition}\label{condition 7 unitarity of transition operator}
    The operator $W^{\rm w}$ of passing from $\mathcal H$ to $\mathcal H^{\rm w}$ is unitary.
\end{Condition}

We consider the space $\mathcal H^{\rm w}$ as the model space. The operator $W^{\rm w}$ defines the unitary copy $W^{\rm w}L_0^*{W^{\rm w}}^*$ of the operator $L_0^*$ which acts in $\mathcal H^{\rm w}$. Since for each $u\in\mathcal U_{L_0}$ there exists a control $h\in\mathcal M$ and $T\geqslant 0$ such that $u=u^h(T)$, we can write
$$
    L_0^*u=L_0^*u^h(T)=-u^h_{tt}(T)=-u^{h_{tt}}(T).
$$
The graph of the unitary image of the wave part of the operator $L_0^*$ can be defined via smooth waves:
\begin{multline*}
    {\rm Graph}\left(W^{\rm w}L_0^*|_{{\cal U}_{L_0}}{W^{\rm w}}^*\right)=\left\{(W^{\rm w}u,W^{\rm w}L_0^*u),u\in\mathcal U_{L_0}\right\}
    \\
    =\left\{(W^{\rm w}u^h(T),-W^{\rm w}u^{h_{tt}}(T)),h\in{\mathcal M},T\geqslant 0\right\}.
\end{multline*}
This way of constructing the wave model is available to the ``outer observer'' who can apply different controls and draw graphs.

\subsubsection{The coordinate representation}
If defect indices of the operator $L_0$ are finite, then under additional assumptions one can define coordinates in spaces of values $\mathcal U_{L_0,\omega}^{\rm w}$ and pass to the wave model, where the operator is represented as a differential operator acting in a space of square integrable functions.

\begin{Condition}\label{condition 8 basis}
    The operator $L_0$ has defect indices $(n,n)$, $n<\infty$. The subspace ${\rm Ker\,}L_0^*$ lies in $\mathcal U_{L_0}$. There exists a basis $e_1,e_2,...,e_n$ in ${\rm Ker\,}L_0^*$ and a set $\Omega_{L_0}^{0}\subseteq\Omega_{L_0}^{e}$, atoms of which form a complete system and for which $\mu(\Omega_{L_0}\backslash\Omega_{L_0}^0)=0$, such that for every $\omega\in\Omega_{L_0}^{0}$ the elements $[e_1](\omega),[e_2](\omega),...,[e_n](\omega)$ form a basis in the space of values $\mathcal U_{L_0,\omega}^{\rm w}$.
\end{Condition}

For atoms $\omega\in\Omega_{L_0}$ and smooth waves $u\in\mathcal U_{L_0}$ elements $[u](\omega)$ can be decomposed over the basis $[e_1](\omega),[e_2](\omega),...,[e_n](\omega)$. Coefficients of this decomposition can be found from the limit
$$
    \hat u(\omega):=\begin{pmatrix}
      \langle u,e_1\rangle_{\omega} \\
      \langle u,e_2\rangle_{\omega} \\
      \vdots \\
      \langle u,e_n\rangle_{\omega} \\
    \end{pmatrix}
    =
    \lim_{t\to+0}\frac1{(P_{\omega(t)}e,e)}
    \begin{pmatrix}
      (P_{\omega(t)}u,e_1) \\
      (P_{\omega(t)}u,e_2) \\
      \vdots \\
      (P_{\omega(t)}u,e_n) \\
    \end{pmatrix}
$$
and the Gram matrix
\begin{multline*}
    G(\omega)=
    \begin{pmatrix}
      \langle e_1,e_1\rangle_{\omega} & \langle e_2,e_1\rangle_{\omega} & \cdots & \langle e_n,e_1\rangle_{\omega} \\
      \langle e_1,e_2\rangle_{\omega} & \langle e_2,e_2\rangle_{\omega} & \cdots & \langle e_n,e_2\rangle_{\omega} \\
      \vdots & \vdots & \ddots & \vdots \\
      \langle e_1,e_n\rangle_{\omega} & \langle e_2,e_n\rangle_{\omega} & \cdots & \langle e_n,e_n\rangle_{\omega} \\
    \end{pmatrix}
    \\
    =\lim_{t\to+0}\frac1{(P_{\omega(t)}e,e)}
    \begin{pmatrix}
      (P_{\omega(t)}e_1,e_1) & (P_{\omega(t)}e_2,e_1) & \cdots & (P_{\omega(t)}e_n,e_1) \\
      (P_{\omega(t)}e_1,e_2) & (P_{\omega(t)}e_2,e_2) & \cdots & (P_{\omega(t)}e_n,e_2) \\
      \vdots & \vdots & \ddots & \vdots \\
      (P_{\omega(t)}e_1,e_n) & (P_{\omega(t)}e_2,e_n) & \cdots & (P_{\omega(t)}e_n,e_n) \\
    \end{pmatrix};
\end{multline*}
this information is available to the ``outer observer''. It is easier, however, to take in the coordinate representation $\hat u(\omega)$ instead of these coefficients as values at $\omega$. In this way we obtain the model of the wave part of the operator $L_0^*$ in the space
$$
    \mathcal H^{\rm c}:=L_2(\Omega_{L_0},\mu,\mathbb C^n)
$$
of the coordinate representation which we also call the wave model. In a perfect situation one can define on $\Omega_{L_0}$ a manifold structure or even global coordinates. This takes place for the Laplace operator on a compact Riemannian manifold \cite{Belishev-2013}, for positive-definite Schr\"odinger operator on the half-line \cite{Belishev-Simonov-2017}, and in our case.

\section{Sturm--Liouville operator on an interval}
Let us look at realization of the abstract scheme for the Sturm--Liouville operator on an interval.

\subsection{The operator $L_0$}
Let $0<l<\infty$, ${\cal H} = L_2(0,l)$. The operator $L_0$ is defined on the domain
\begin{equation}\label{Dom}
    {\rm Dom\,} L_0 =\left\{u\in H^2(0,l): u(0)=u'(0)=u(l)=u'(l)=0\right\}
\end{equation}
by the differential expression
\begin{equation}\label{differential expression}
    L_0u:=-u''+qu,
\end{equation}
where $q\in C^{\infty}[0,l]$ is a smooth function such that the operator $L_0$ is positive-definite. Such an operator is symmetric and has the defect indices $(2,2)$. Its adjoint $L_0^*$ is defined by the same differential expression on the domain
$$
    {\rm Dom\,} L_0^* =H^2(0,l).
$$
The Friedrichs extension $L$ of $L_0$ is defined on the domain
$$
    {\rm Dom\,} L =\left\{u\in H^2(0,l): u(0)=u(l)=0\right\}.
$$

\subsection{The Green's system}
To describe the subspace ${\mathcal K}={\rm Ker\,}L_0^*$ define two solutions of the equation $-u''+qu=0$. Denote by $\phi_0$ the solution of this equation with the initial data $\phi_0(0)=0$, $\phi_0'(0)=1$ and by $\phi_l$ the solution with the data $\phi_l(l)=0$, $\phi_l'(l)=1$. Since the operator $L$ is positive-definite, $0$ is not its eigenvalue and these functions cannot be proportional. Therefore they form a basis in $\mathcal K$.

Let us write out the Vishik's decomposition for $u\in{\rm Dom\,}L_0$. Let
$$
    \eta_0:=L^{-1}\phi_0,\ \eta_l:=L^{-1}\phi_l.
$$

\begin{Lemma}
    In the decomposition of $u\in{\rm Dom\,}L_0^*$
    $$
        u=u_0+L^{-1}g_u+h_u,
    $$
    the elements $g_u,h_u\in\mathcal K$ are given by the formulas
    \begin{multline}\label{g-u}
        g_u=\frac{1}{\eta_0'(0)\eta_l'(l)-\eta_l'(0)\eta_0'(l)}
        \\
        \times\left\{\bigg[\eta_l'(l)\bigg(u'(0)-\frac{u(l)}{\phi_0(l)}-\frac{u(0)}{\phi_l(0)}\phi_l'(0)\bigg)\right.
        -\eta_l'(0)\bigg(u'(l)-\frac{u(l)}{\phi_0(l)}\phi_0'(l)-\frac{u(0)}{\phi_l(0)}\bigg)\bigg]\eta_0
        \\
        +\bigg[\eta_0'(0)\bigg(u'(l)-\frac{u(l)}{\phi_0(l)}\phi_0'(l)-\frac{u(0)}{\phi_l(0)}\bigg)
        \left.\left.-\eta_0'(l)\left(u'(0)-\frac{u(l)}{\phi_0(l)}-\frac{u(0)}{\phi_l(0)}\phi_l'(0)\right)\right]\eta_l\right\},
    \end{multline}
    \begin{equation}\label{h-u}
        h_u=\frac{u(l)}{\phi_0(l)}\phi_0+\frac{u(0)}{\phi_l(0)}\phi_l.
    \end{equation}
\end{Lemma}

\begin{proof}
Since $u_0(0)=u_0'(0)=u_0(l)=u_0'(l)=0$ and $(L^{-1}g_u)(0)=(L^{-1}g_u)(l)=0$, we should find the coefficients in the equalities
\begin{equation}\label{h-u,g-u}
    h_u=c_0\phi_0+c_l\phi_l,\ g_u=d_0\phi_0+d_l\phi_l,
\end{equation}
such that
\begin{align*}
    u(0)&=h_u(0),
    \\
    u(l)&=h_u(l),
    \\
    u'(0)&=(L^{-1}g_u)'(0)+h_u'(0),
    \\
    u'(l)&=(L^{-1}g_u)'(l)+h_u'(l).
\end{align*}
Substituting here \eqref{h-u,g-u} and $L^{-1}g_u=d_0\eta_0+d_l\eta_l$ and taking into account initial data for the solutions $\phi$, we find the coefficients $c_0,c_l,d_0$ and $d_l$ and arrive at the formulas \eqref{g-u}  \eqref{h-u}.
\end{proof}

We get from the lemma and \eqref{**}:
\begin{equation}\label{Gamma-1}
    \Gamma_1u=-\frac{u(l)}{\phi_0(l)}\phi_0-\frac{u(0)}{\phi_l(0)}\phi_l,
\end{equation}
\begin{multline}\label{Gamma-2}
    \Gamma_2u=\frac{1}{\eta_0'(0)\eta_l'(l)-\eta_l'(0)\eta_0'(l)}
    \\
    \times\left\{\bigg[u(0)\frac{\eta_l'(0)-\eta_l'(l)\phi_l'(0)}{\phi_l(0)}-u(l)\frac{\eta_l'(l)-\eta_l'(0)\phi_0'(l)}{\phi_0(l)}\right.+u'(0)\eta_l'(l)-u'(l)\eta_l'(0)\bigg]\phi_0
    \\
    -\bigg[u(0)\frac{\eta_0'(0)+\eta_0'(l)\phi_l'(0)}{\phi_l(0)}\left.-u(l)\frac{\eta_0'(l)-\eta_0'(0)\phi_0'(l)}{\phi_0(l)}+u'(0)\eta_0'(l)-u'(l)\eta_0'(0)\bigg]\phi_l\right\}.
\end{multline}
The spaces $\mathcal H=L_2(0,l)$, $\mathcal K = \{c_0\phi_0+c_l\phi_l,c_0,c_l\in\mathbb C\}$ and the operators $L_0,\Gamma_1,\Gamma_2$ defined by \eqref{Dom}, \eqref{differential expression}, \eqref{Gamma-1}, and \eqref{Gamma-2} form the Green's system $\mathfrak G_{L_0}$, which canonically corresponds to the operator $L_0$.

\subsection{The system with boundary control}
Consider the system \eqref{alpha1}--\eqref{alpha3} in our case. The boundary control $h(t)\in\mathcal K$ can be written in the form
$$
    h(t)=-\frac{f_l(t)}{\phi_0(l)}\phi_0-\frac{f_0(t)}{\phi_l(0)}\phi_l,
$$
where the functions $f_0(t)$ and $f_l(t)$ are taken from the class
\begin{equation}\label{linal cal M SL}
    \dot{\mathcal M}=\{f \in C^{\infty}\left[0,\infty\right):\ {\rm supp\,}f \subset (0,\infty)\}.
\end{equation}
Then the system \eqref{alpha1}--\eqref{alpha3} takes the form of the initial-boundary value problem
\begin{align*}
    & u_{tt}-u_{xx}+qu = 0,  && x\in(0,l),t>0,
    \\
    & u|_{t=0}=u_t|_{t=0}=0, && x\in[0,l],
    \\
    & u|_{x=0} = f_0(t), && t\geqslant 0,
    \\
    & u|_{x=l} = f_l(t), && t\geqslant 0.
\end{align*}
The solution of such a problem for $t\leqslant l$ is given by the formula
\begin{multline}\label{u^f repres}
    u^f(x,t)=f_0(t-x)+f_l(t-l+x)
    \\
    +\int_x^t w_0(x,s)f_0(t-s)ds+\int_{l-x}^t w_l(l-x,s)f_l(t-s)ds,
\end{multline}
where the functions $f_0$ and $f_l$ are assumed to be zero on the negative half-line, the functions $w_0(x,t)$ and $w_l(x,t)$ are defined for $0\leqslant x\leqslant t\leqslant l$ and are smooth.

\subsubsection{Controllability of the system $\alpha_{L_0}$}\label{ssec Controllability}
Let us find reachable sets of the system $\alpha_{L_0}$.

\begin{Lemma}\label{Lemma controll SL}
    \begin{equation}\label{UT SL}
        \mathcal U_{L_0}^t=
        \left\{
        \begin{array}{ll}
            \{u\in C^{\infty}[0,l]:\ {\rm supp}\, u\subset[0,t)\cup(l-t,l]\},&t\leqslant \frac l2,
            \\
            C^{\infty}[0,l],&t>\frac l2.
        \end{array}
        \right.
    \end{equation}
\end{Lemma}

\begin{proof}
Let $t\leqslant \frac l2$. One can see from the expression \eqref{u^f repres} that for $f_0,f_l\in\dot{\mathcal M}$ the solution $u^h(\cdot,t)$ belongs to $C^{\infty}[0,l]$. It also follows that its support is contained in $[0,t)\cup(l-t,l]$. Thus
$$
    \mathcal U_{L_0}^t\subseteq\{u\in C^{\infty}[0,l]:\ {\rm supp}\, u\subseteq[0,t)\cup(l-t,l]\}.
$$
To prove the inverse inclusion, take $u$ from the right-hand side and show that $u(x)=u^f(x,t)$. Let us represent $u$ in the form
$$
    u=u_0+u_l,\ u_0,u_l\in C^{\infty}[0,l], {\rm supp}\,u_0\subseteq[0,t),{\rm supp}\,u_l\subseteq(l-t,l].
$$
Divide the equation $u(x)=u^f(x,t)$, according to \eqref{u^f repres}, into two parts as follows:
$$
    f_0(t-x)+\int_x^t w_0(x,s)f_0(t-s)ds=u_0(x),
$$
$$
    f_l(t-l+x)+\int_{l-x}^t w_l(l-x,s)f_l(t-s)ds=u_l(x).
$$
These are Volterra equations of the second kind on the interval $(0,l)$, they have solutions from the same classes, to which their right-hand sides belong (taking into account change of the variable; ${\rm supp}\,f_0, {\rm supp}\,f_l\subseteq(0,t]$, they can be continued to $\dot{\mathcal M}$, which will not affect the equality $u(x)=u^f(x,t)$). Thus the first assertion of the lemma is proved.

Let $\frac l2<t\leqslant l$ and $u\in C^{\infty}[0,l]$. There exists a function $u_0\in C^{\infty}[0,l]$ such that $u_0|_{[0,\frac l2]}=u|_{[0,\frac l2]}$ and ${\rm supp}\,u_0\subseteq [0,t)$. Take $u_l=u-u_0$. Then $u_l\in C^{\infty}[0,l]$ and ${\rm supp}\,u_l\subseteq[\frac l2,l]\subseteq(l-t,l]$. By the same argument as in the first part of the proof we obtain controls $f_0,f_l\in\dot{\mathcal M}$ for which $u(x)=u^f(x,t)$. Consequently, $C^{\infty}[0,l]\subseteq \mathcal U_{L_0}^t$. From \eqref{u^f repres} it follows that $\mathcal U_{L_0}^t\subseteq C^{\infty}[0,l]$.

For $t>l$ the inclusion $\mathcal U_{L_0}^l=C^{\infty}[0,l]\subseteq \mathcal U_{L_0}^t$ holds owing to monotonicity of reachable sets, and the inverse inclusion $\mathcal U_{L_0}^t\subseteq C^{\infty}[0,l]$ is always true. Thus the lemma is proved.
\end{proof}

The system $\alpha_{L_0}$ is controllable, since $\overline{\mathcal U}_{L_0}=L_2(0,l)=\mathcal H$, and this also follows from the fact that the operator $L_0$ is completely non-selfadjoint. Closure of $C^{\infty}[0,l]$ in the graph norm of the operator $L_0^*$ is the Sobolev space $H^2(0,l)={\rm Dom}\, L_0^*$, therefore the wave part of the operator $L_0^*$ (which is $\overline{L_0^*|_{\mathcal U_{L_0}}}$) coincides with $L_0^*$.

\subsection{The wave spectrum}
We turn to constructing the wave spectrum of the operator $L_0$. For this we have already found the family of reachable subspaces $\overline{\mathcal U^t_{L_0}}=L_2((0,t)\cup(l-t,l))$. Now we have to find out how the wave isotony $I_L$ acts.

For a set $E\subset[0,l]$ denote by $E^t$ its metric neighborhood in $[0,l]$:
$$
    E^t=\{x\in[0,l]:\ {\rm dist}(x,E)<t\},\ t>0,
$$
$$
    {\rm dist\,}(x,E):=\underset{y\in E}{\rm inf}\,{\rm dist}(x,y).
$$
For $t=0$ we take $E^t=E$.

\begin{Lemma}\label{Lemma Basic}
    For $0\leqslant a<b\leqslant l$ and $t\geqslant 0$ the following holds:
    \begin{equation}\label{isotony on interval}
        I^t_L(L_2(a,b))=L_2((a,b)^t).
    \end{equation}
\end{Lemma}

\begin{Remark}
    We identify spaces $L_2(a,b)$ with the subspaces of $L_2(0,l)$ which consist of functions that vanish a.e. outside $(a,b)$.
\end{Remark}

\begin{proof}
The system \eqref{dual 1}--\eqref{dual 2} can be written in the form of the initial-boundary value problem
\begin{align}
    \label{dual 1 SL} & v_{tt}-v_{xx}+qv=g, && x\in(0,l),t>0,
    \\
    \label{dual 2 SL} & v|_{t=0}=v_t|_{t=0}=0, && x\in[0,l],
    \\
    \label{dual 3 SL} & v|_{x=0}=v|_{x=l}=0, && t\geqslant0.
\end{align}
with the right-hand side $g(x,t)$ from the corresponding class.

An argument analogous to the proof of Lemma 2 from \cite{Belishev-Simonov-2017}, which is based on the fact of finiteness of the domain of influence for the hyperbolic equation \eqref{dual 1 SL}, leads to the inclusion ${\mathcal V}^t_L(L_2(a,b))\subseteq L_2((a,b)^t)$ and hence to $I_L^t(L_2(a,b))\subseteq L_2((a,b)^t)$.

Consider the conjugate problem
\begin{align}
    \label{d-dual 1 SL} & w_{tt}-w_{xx}+qw=0, && x\in(0,l),t\in(0,T),
    \\
    \label{d-dual 2 SL} & w|_{t=T}=0,\ w_t|_{t=T}=y, && x\in[0,l],
    \\
    \label{d-dual 3 SL} & w|_{x=0}=w|_{x=l}=0, && t\in[0,T].
\end{align}
For $g\in C^{\infty}_0((0,l)\times(0,\infty))$ and $y \in L_2(0,l)$ the duality relation
\begin{equation}\label{duality relation}
    \int_0^l\int_0^Tg(x,t)w^y(x,t)dxdt=-\int_0^lv^g(x,T)y(x)dx
\end{equation}
holds. The odd continuation of the solution $w^y$ solves the problem
\begin{align}
    \label{dd-dual 1 SL} & w_{tt}-w_{xx}+qw=0, && x\in(0,l),t\in(0,2T),
    \\
    \label{dd-dual 2 SL} & w|_{t=T}=0,\ w_t|_{t=T}=y, && x\in[0,l],
    \\
    \label{dd-dual 3 SL} & w|_{x=0}=w|_{x=l}=0, && t\in[0,2T].
\end{align}
(note that both $w^y$ and $w^y_t$ retain continuity). If there exists $y \in L_2((a,b)^T)\ominus {\mathcal V}^T_L(L_2(a,b))$, then an argument analogous to the proof of Lemma 2 from \cite{Belishev-Simonov-2017} leads to $w^y=0$, from which it follows that $y$ can be only zero. Therefore ${\mathcal V}^t_L(L_2(a,b))$ is dense in $L_2((a,b)^t)$. Thus we proved that $I_L^t(L_2(a,b))=L_2((a,b)^t)$.
\end{proof}

Let us call a set $E \subseteq[0,l]$ {\it elementary}, if
$$
    E=\bigcup_{k=1}^{n(E)}(a_k,b_k),
$$
where $0\leqslant a_1<b_1<a_2<b_2<...<a_{n(E)}<b_{n(E)}\leqslant l$ and if the set $E$ is symmetric with respect to the middle of the interval $(0,l)$. Let $\mathcal E[0,l]$ be the family of all elementary sets. Obviously, if $E\in\mathcal E[0,l]$, then $E^t\in\mathcal E[0,l]$ for every $t\geqslant0$. We will also call the subspaces $L_2(E)$, $E\in\mathcal E[0,l]$, {\it elementary}. The family of elementary subspaces forms the lattice ${\mathfrak L}_{\mathcal E[0,l]}\subseteq\mathfrak L({\mathcal H})$.

\begin{Lemma}\label{Lemma Basic 2}
    For every $E\in\mathcal E[0,l]$ one has $I_L^t(L_2(E))=L_2(E^t)$.
\end{Lemma}

\begin{proof}
By isotonicity,
$$
    L_2((a_k,b_k)^t)=I^t_L(L_2(a_k,b_k))\subseteq I_L^t(L_2(E))
$$
for every $k$, and thus $L_2(E^t)\subseteq I_L^t(L_2(E))$. Using the same argument as in the proof of Lemma \ref{Lemma Basic}, we arrive to
$$
    I_L^t(L_2(E))=\overline{\mathcal V^t_L(L_2(E))}=L_2(E^t).
$$
\end{proof}

The lattice $\mathfrak L_{\mathcal E[0,l]}$ is invariant under the wave isotony $I_L$ and contains all the subspaces of the form $L_2((0,t)\cup(l-t,l))$, i.e., all reachable subspaces. Therefore $\mathfrak L_{L_0}=\mathfrak L_{\mathcal E[0,l]}$.

Let $m$ denote the Lebesgue measure, $\mathcal B$ the Borel sigma-algebra on the segment $[0,l]$, $\mathfrak L_{\mathcal B}$ the corresponding lattice of subspaces,
$$
    \mathfrak L_{\mathcal B}:=\{L_2(E),\ E\in\mathcal B\}\subseteq\mathfrak L({\mathcal H}),
$$
$E\triangle F=(E\setminus F)\cup(F\setminus E)$ the symmetric difference of sets.

\begin{Lemma}\label{Lemma convergence in L(H)}
    Let $\{E_n\}_{n\in\mathbb N}$ be a sequence of sets from $\mathcal B$ and $E\in\mathcal B$. Then convergence $L_2(E_n)\underset{n\rightarrow\infty}\longrightarrow L_2(E)$ in the topology of $\mathfrak L({\mathcal H})$ is equivalent to $m(E_n\triangle E)\underset{n\to\infty}\longrightarrow0$.
\end{Lemma}

The proof of the lemma repeats the proof of Lemma 4 from \cite{Belishev-Simonov-2017} almost literally.

\begin{Lemma}\label{Lemma closure of L-L-0 is in L-Leb}
    The closure of the lattice $\mathfrak L_{L_0}$ in the topology of $\mathfrak L({\mathcal H})$ is a subset of the lattice $\mathfrak L_{\mathcal B}$:
    $$
        \overline{\mathfrak L}_{L_0}\subseteq\mathfrak L_{\mathcal B}.
    $$
\end{Lemma}

\begin{proof}
Let a sequence $L_2(E_n)$ of subspaces from $\mathfrak L_{L_0}=\mathfrak L_{\mathcal E[0,l]}$ be fundamental in $\mathfrak L({\mathcal H})$. Let us prove that there exists $E\in\mathcal B$ such that $L_2(E_n)\overset{\mathfrak L({\mathcal H})}\longrightarrow L_2(E)$. By Lemma \ref{Lemma convergence in L(H)}, convergence means that $m(E_n\triangle E)\rightarrow0$. The symmetric difference is a pseudometric in $\mathcal B$ and, after factorization with respect to the equivalence relation of the form $E\sim F$, if $m(E\triangle F)=0$, we get $\mathcal B/_{\sim}$, \cite{Kolmogorov-Fomin-1989}. Thus there exists a measurable set $E\subseteq[0,l]$ such that $m(E_n\triangle E)\rightarrow0$, and by Lemma \ref{Lemma convergence in L(H)} this means that $L_2(E_n)\overset{\mathfrak L({\mathcal H})}\longrightarrow L_2(E)$.
\end{proof}

\begin{Remark}
The set $E$ should be symmetric (up to a set of zero measure) with respect to the middle of the interval $(0,l)$, and therefore $\overline{\mathfrak L}_{L_0}\neq\mathfrak L_{\mathcal B}$.
\end{Remark}

\begin{Corollary}
    Functions of the family $[I_L\mathfrak L_{L_0}]_{\rm seq}$ are isotonic and take values in $\mathfrak L_{\mathcal B}$.
\end{Corollary}

Consider the metric space $\mathcal B/_{\sim}$ of equivalence classes of measurable sets with the distance $\rho(E_{\sim},F_{\sim})=m(E\triangle F)$ and for each $t>0$ consider the following sets in it:
\begin{multline*}
    \mathcal E_{>t}:=\{E^t,\ E\in\mathcal E[0,l]\}
    \\=\{E\in\mathcal E[0,l]:\ b_1>t,\text{ if }a_1=0,\ b_1-a_1>2t,\text{ if }a_1\neq0,
    \\\text{ and }b_k-a_k>2t,\forall k=2,...,n-1\},
\end{multline*}
\begin{multline*}
    \mathcal E_{\geqslant t}:=\{E\in\mathcal E[0,l]:\ b_1\geqslant t,\text{ if }a_1=0,\ b_1-a_1\geqslant2t,\text{ if }a_1\neq0,
    \\\text{ and }b_k-a_k\geqslant2t,\forall k=2,...,n-1\}.
\end{multline*}
Recall that elementary sets are symmetric with respect to the middle of the interval $(0,l)$.

\begin{Lemma}\label{Lemma closure >t}
    Closure of $(\mathcal E_{>t})_{\sim}$ in the metric of $\mathcal B/_{\sim}$ is a subset of $(\mathcal E_{\geqslant t})_{\sim}$.
\end{Lemma}

\begin{proof}
Let $\{E_n\}_{n\in\mathbb N}$ be a sequence from $\mathcal E_{>t}$ such that $(E_n)_{\sim}\overset{\mathcal B/_{\sim}}\longrightarrow E_{\sim}\in\mathcal B /_{\sim}$. Each of the sets $E_n$ contains no more than $\frac lt$ component intervals. One can choose a subsequence $\{E_{n_j}\}_{j\in\mathbb N}$ of sets which all contain the same number of component intervals. Denote this number by $N$.
\begin{figure}[h]
    \includegraphics[width=\textwidth]{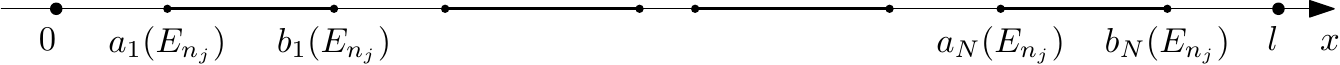}
    \caption{The set $E_{n_j}$}
    \label{fig1_limiting_interval}
\end{figure}
One can choose a subsequence $\{E_{n_{j_l}}\}_{l\in\mathbb N}$ of sets such that all the endpoints of the component intervals $\{a(E_{n_{j_l}})\}_{l\in\mathbb N}$, $\{b(E_{n_{j_l}})\}_{l\in\mathbb N}$     converge to some numbers $0\leqslant a_1\leqslant b_1\leqslant a_2\leqslant b_2\leqslant ...\leqslant a_N\leqslant b_N\leqslant l$ (see Fig. \ref{fig1_limiting_interval} and \ref{fig2_limit_interval}).
\begin{figure}[h]
    \includegraphics[width=\textwidth]{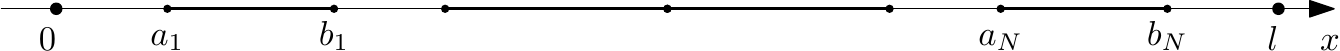}
    \caption{The set $E_{\infty}$}
    \label{fig2_limit_interval}
\end{figure}
In this way we obtain the set
$$
    E_{\infty}:=\bigcup_{k=1}^{N}(a_k,b_k).
$$
\begin{figure}[h]
    \begin{center}
    \includegraphics[width=0.7\textwidth]{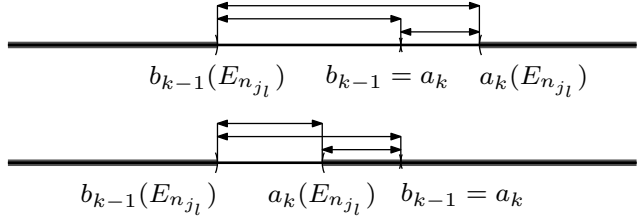}
    \caption{Estimating the measure of the symmetric difference}
    \label{fig3_estimate}
    \end{center}
\end{figure}
It is easy to see the following estimate (see Fig. \ref{fig3_estimate}):
$$
    m(E_{n_{j_l}}\triangle E_{\infty})\leqslant\sum_{k=1}^N(|a_k(E_{n_{j_l}})-a_k|+|b_k(E_{n_{j_l}})-b_k|).
$$
Consequently, $m(E_{n_{j_l}}\triangle E_{\infty})\underset{l\to\infty}\longrightarrow0$. This means that $E_{\sim}=(E_{\infty})_{\sim}$. Since $E_{\infty}\in\mathcal E_{\geqslant t}$, the sequence $\{E_n\}_{n\in\mathbb N}$ was an arbitrary convergent sequence from $\mathcal E_{>t}$, the lemma is proved.
\end{proof}

\begin{Lemma}\label{Lemma inclusion}
    Let $E,F\subseteq[0,l]$, $\{E_n\}_{n\in\mathbb N}$, $\{F_n\}_{n\in\mathbb N}$ be sequences of subsets of the segment $[0,l]$. Let $m(E_n\triangle E)\to0$ and $m(F_n\triangle F)\to0$ as $n\to\infty$. If $E_n\subseteq F_n$ for every $n$, then $m(E\setminus F)=0$.
\end{Lemma}

\begin{proof}
We have:
\begin{multline*}
    E\setminus F\subseteq(E\cup E_n)\setminus F=(E_n\cup(E\setminus E_n))\setminus F
    =(E_n\setminus F)\cup((E\setminus E_n)\setminus F)
    \\
    \subseteq(F_n\setminus F)\cup(E_n\setminus E)\subseteq(F_n\triangle F)\cup(E_n\triangle E).
\end{multline*}
From this we immediately get the assertion of the lemma.
\end{proof}

For $x\in[0,\frac l2]$ denote
$$
    \omega_x(t):=L_2((\{x\}\cup\{l-x\})^t),
$$
then $\omega_x\in[I_L\mathfrak L_{L_0}]_{\rm seq}$. Indeed, $L_2(((\{x\}\cup\{l-x\})^\frac 1n)^t)\in I_L\mathfrak L_{L_0}$ for every $n$ and $L_2(((\{x\}\cup\{l-x\})^\frac 1n)^t)\underset{n\to\infty}\longrightarrow\omega_x$ in $\mathfrak F(\mathcal H)$.

\begin{Lemma}\label{Lemma main}
For every nonzero $\omega\in[I_L\mathfrak L_{L_0}]_{\rm seq}$ there exists $x\in[0,l]$ such that $\omega_x\leqslant\omega$.
\end{Lemma}

\begin{proof}
Let $\omega\in[I_L\mathfrak L_{L_0}]_{\rm seq}$. This means that there exists a sequence $\{E_n\}_{n\in\mathbb N}$ of elementary sets such that $I^t_L(L_2(E_n))=L_2(E^t_n)\to\omega(t)$ in $\mathfrak L({\mathcal H})$ for every $t\geqslant0$. By Lemma \ref{Lemma closure of L-L-0 is in L-Leb} there exist measurable sets $E(t)\subseteq[0,l]$ such that $\omega(t)=L_2(E(t))$, and by Lemma \ref{Lemma closure >t} $E(t)\in\mathcal E_{\geqslant t}$.

If $E(t)=[0,l]$ for every $t>0$, then the assertion of the lemma holds: every element $\omega_x$, $x\in[0,\frac l2]$, satisfies $\omega_x\leqslant\omega$. Assume that there exists $t_0>0$ such that $E(t_0)\neq[0,l]$. Then for the right endpoint of the first interval the inequality $b_1(E(t_0))>t_0$ holds (two cases are possible, see Fig. \ref{fig4_case1} and \ref{fig5_case2}).
\begin{figure}[h]
    \begin{center}
    \includegraphics[width=0.8\textwidth]{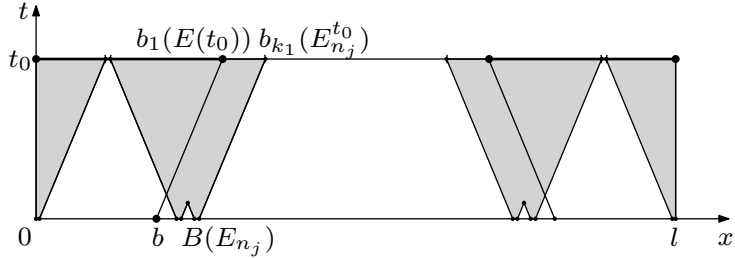}
    \caption{The first case}
    \label{fig4_case1}
    \end{center}
\end{figure}
\begin{figure}[h]
    \begin{center}
    \includegraphics[width=0.8\textwidth]{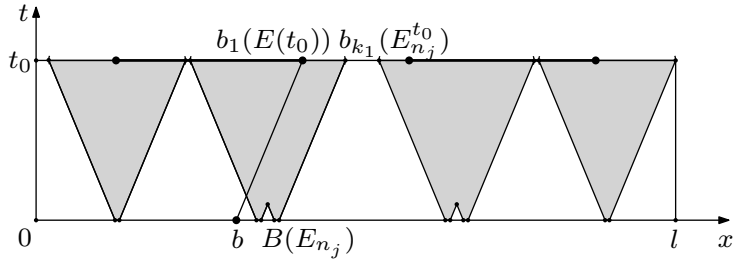}
    \caption{The second case}
    \label{fig5_case2}
    \end{center}
\end{figure}
The set $E_n^{t_0}$ contains the finite number of component intervals, there exists a sequence $\{E_{n_j}^{t_0}\}_{j\in\mathbb N}$ of sets which all contain the same number of component intervals, and the endpoints of these intervals have limits. These limits can be either endpoints of component intervals of the set $E(t_0)$ or inner points of this set.

The point $b_1(E(t_0))$ is the limit of some sequence $\{b_{k_1}(E_{n_j}^{t_0})\}_{j\in\mathbb N}$ of the right endpoints of the component intervals with some fixed number $k_1$ of the sets $E_{n_j}^{t_0}$. Denote
$$
    b:=b_1(E(t_0))-t_0,\ B(E_{n_j}):=b_{k_1}(E_{n_j}^{t_0})-t_0
$$
(see Fig. \ref{fig4_case1} and \ref{fig5_case2}). Then $B(E_{n_j})\to b$ as $j\to\infty$. The sets $E_{n_j}^t$ contain $\{B(E_{n_j})\}^t$ for every $t>0$. Since $E_{n_j}^t\to E(t)$, $m(E_{n_j}^t\triangle E(t))\to0$ and $m((\{B(E_{n_j})\}^t)\triangle(\{b\}^t))\to0$ as $j\to\infty$, by Lemma \ref{Lemma inclusion} we obtain that $\{b\}^t\subseteq E(t)$ up to a set of measure zero for every $t>0$. Therefore $\omega_b\leqslant\omega$.
\end{proof}

Now we can describe the wave spectrum of the operator $L_0$.

\begin{Theorem}\label{Th 1}
$$
\Omega_{L_0}=\left\{\omega_x,\ x\in\left[0,\frac l2\right]\right\}.
$$
\end{Theorem}

\begin{proof}
Let $\omega\in\Omega_{L_0}$. Since $\Omega_{L_0}\subseteq[I_L\mathfrak L_{L_0}]_{\rm seq}$, by Lemma \ref{Lemma main} there exists $x\in[0,\frac l2]$ such that $\omega_x\leqslant\omega$. Since $\omega$ is an atom of the set $[I_L\mathfrak L_{L_0}]_{\rm seq}$, it should be that $\omega_x=\omega$. Hence $\Omega_{L_0}\subseteq\{\omega_x,\ x\in[0,\frac l2]\}$.

Let us prove the inverse inclusion. Let $x\in[0,\frac l2]$ and let there exist an element $\omega\in[I_L\mathfrak L_{L_0}]_{\rm seq}$ such that $\omega<\omega_x$. By Lemma \ref{Lemma main} there exists $\tilde x\in[0,\frac l2]$ such that $\omega_{\tilde x}\leqslant\omega$. This means that $\omega_{\tilde x}<\omega_x$. But this cannot happen: for $\tilde x=x$ we have $\omega_{\tilde x}=\omega_x$, while for $\tilde x\neq x$ and $t<|\tilde x-x|$ we have $\omega_x(t)\cap\omega_{\tilde x}(t)=\{0\}$ (see Fig. \ref{fig6_x_and_x-tilde}), which contradicts the inequality $\omega_{\tilde x}<\omega_x$. \begin{figure}[h]
    \begin{center}
    \includegraphics[width=0.7\textwidth]{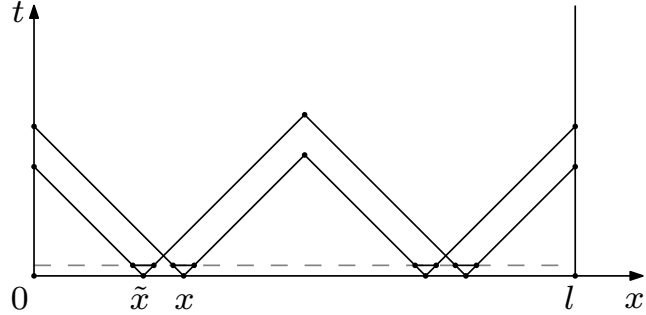}
    \caption{$\omega_{\tilde x}$ and $\omega_x$}
    \label{fig6_x_and_x-tilde}
    \end{center}
\end{figure}
Therefore such $\omega$ does not exist and $\omega_x$ is an atom. Hence $\{\omega_x,\ x\in[0,\frac l2]\}\subseteq\Omega_{L_0}$ and the theorem is proved.
\end{proof}

Denote by $\beta$ the bijection between $[0,\frac l2]$ and $\Omega_{L_0}$ established by Theorem \ref{Th 1}, $\beta: x\mapsto\omega_x$. Let us denote also $x_{\omega}:=\beta^{-1}(\omega)$, $\omega\in\Omega_{L_0}$, $E_x(t):=(\{x\}\cup\{l-x\})^t$ and $f_{\omega}(x):={\rm dist}\,(x,(\{x_{\omega}\}\cup\{l-x_{\omega}\})^t)$. Note that
\begin{equation}\label{E-x-omega}
    E_{x_{\omega}}(t)=\{y\in(0,l):\ f_{\omega}(y)<t\}
\end{equation}
(see Fig. \ref{fig7_f-omega}).
\begin{figure}[h]
    \begin{center}
    \includegraphics[width=0.7\textwidth]{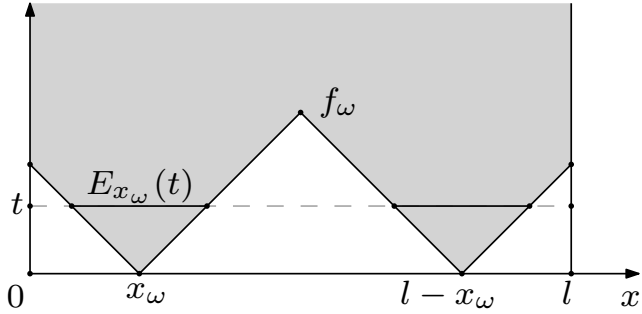}
    \caption{The set $E_{x_{\omega}}(t)$ and the graph of the function $f_{\omega}$}
    \label{fig7_f-omega}
    \end{center}
\end{figure}

\begin{Lemma}\label{Lemma eikonal}
    Let $\omega\in\Omega_{L_0}$. Then the family of projections
    $$
        E_{\omega}(t)=
        \left\{
        \begin{array}{ll}
        P_{\omega(t)},\ &t\geqslant0,
        \\
        0,&t<0,
        \end{array}
        \right.
    $$
    is a resolution of the identity in the space $\mathcal H=L_2(0,l)$, and the corresponding eikonal
    $$
        \tau_{\omega}=\int_{\mathbb R}tdE_{\omega}(t)
    $$
    is the operator of multiplication by the function $f_{\omega}$ in $L_2(0,l)$.
\end{Lemma}

\begin{proof}
As one can see from the definition of elements $\omega_x$, for $t>\frac l2$ one has $\omega_x(t)=\mathcal H$, and so $E(t)\overset{s}\to I$ as $t\to+\infty$. Strong left-continuity of functions $P_{\omega_x(t)}=[\chi_{E_x(t)}]$ also takes place. Therefore the family $E(t)$ is indeed a resolution of the identity and defines the (Stieltjes) integral
$\int_{\mathbb R}tdE_{\omega}(t)$. If $M_f$ is the operator of multiplication by the function $f$, $M_f=[f]$, in the space $L_2(\mathbb R,\rho)$ with the measure $\rho$, then the corresponding resolution of the identity is $E(\lambda)=[\chi_{f^{-1}(-\infty,\lambda)}]$. In our case $\rho$ is the Lebesgue measure on the segment $[0,l]$, and for the operator $M_{f_{\omega_x}}=[f_{\omega_x}]$ we get $E(\lambda)=[\chi_{f^{-1}_{\omega_x}(-\infty,\lambda)}]=[\chi_{E_x(t)}]$ by \eqref{E-x-omega}. This means that $E(\lambda)=P_{\omega_x}(\lambda)$ for $\lambda\geqslant0$ and $E(\lambda)=0$ for $\lambda<0$. Since spectral measures are the same, the operators also are, thus $\tau_{\omega}=M_{f_{\omega}}$.
\end{proof}

As one can see, for every $\omega_1,\omega_2\in\Omega_{L_0}$ the distance
$$
    \tau(\omega_1,\omega_2)=\|\tau_{\omega_1}-\tau_{\omega_2}\|=\|[f_{\omega_1}-f_{\omega_2}]\|=|x_{\omega_1}-x_{\omega_2}|
$$
is correctly defined and the wave spectrum becomes a complete metric space. Thus the map $\beta$ is an isometric isomorphism between the segment $[0,\frac l2]$ and the wave spectrum $\Omega_{L_0}$. The ``balls''
$$
    B_r(\omega)=\{\tilde\omega\in\Omega_{L_0}:\ \exists t>0:\ \tilde\omega(t)\neq0,\tilde\omega(t)\subseteq\omega(r)\}
$$
obviously coincide with
$$
    \left\{\omega_{\tilde x},\ \tilde x\in\left[0,\frac l2\right] :\ |\tilde x-x_{\omega}|<r\right\}=\{\tilde\omega\in\Omega_{L_0}:\ \tau(\tilde\omega,\omega)<r\}
$$
(see Fig. \ref{fig8_ball}), i.e., with the balls for the metric $\tau$, so that the ``ball'' topology on the wave spectrum coincides with the topology defined by this metric.
\begin{figure}[h]
    \begin{center}
    \includegraphics[width=0.7\textwidth]{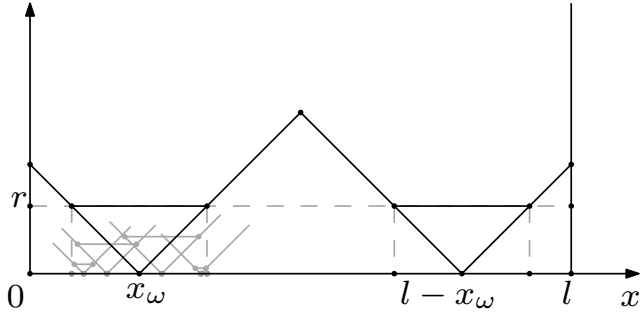}
    \caption{$B_r(\omega)$}
    \label{fig8_ball}
    \end{center}
\end{figure}
From the form of reachable spaces \eqref{UT SL} and the definition of the boundary of the wave spectrum $\partial\Omega_{L_0}$ it follows that in our case
$$
    \partial\Omega_{L_0}=\{\omega_0\}.
$$
The atom $\omega_{\frac l2}$ is not a point of the boundary. Furthermore, the distance from the boundary defines the coordinate
$$
    \tau(\omega):=\tau(\omega,\partial\Omega_{L_0})=x_{\omega},
$$
which parametrizes the wave spectrum for the ``outer observer'' (unlike the isomorphism $\beta$ available only to the ``inner observer'').

\subsection{The wave model}
We begin constructing the wave model of the operator $L_0$ starting with the space of values. The first three Conditions from the abstract part are satisfied, which is obvious since we explicitly know the subspaces $\omega_x(t)$. It is also clear that atoms vanish at zero. To prove existence of a gauge element we need the following standard lemma.

\begin{Lemma}\label{Lemma not many zeros}
A function $u\in{\rm Ker}\,L_0^*$ cannot have more than one zero on the segment $[0,l]$.
\end{Lemma}

\begin{proof}
Let $u\in{\rm Ker}\,L_0^*$. The operator $L$ is positive-definite, and hence its kernel is trivial. Therefore $u$ cannot vanish at both points $0$ and $l$ simultaneously. Assume that $u$ has two zeros, $a$ and $b$, on the segment $[0,l]$, and at least one of them is an inner point. Then $u$ is in the kernel of the Strum--Liouville operator $L_{ab}$ defined on the interval $(a,b)$ by the differential expression $-\frac{d^2}{dx^2}+q(x)$ with the Dirichlet boundary conditions at the points $a$ and $b$. Such an operator is self-adjoint and semibounded from below. Let $l$ and $l_{ab}$ denote the sesquilinear forms that correspond to the operators $L$ and $L_{ab}$. Their domains are $d[l]=\mathring H^1(0,l)$ and $d[l_{ab}]=\mathring H^1(a,b)$. According to the minimax principle \cite{Birman-Solomyak-1980},
$$
    \lambda_1(L_{ab})=\min_{u\in\mathring H^1(a,b)}\frac{(l_{ab}u,u)_{L_2(a,b)}}{\|u\|^2_{L_2(a,b)}}.
$$
If a function $u\in\mathring H^1(a,b)$ is continued by zero to the whole segment $[0,l]$, then one gets the function $\tilde u\in\mathring H^1(0,l)$ and $\|\tilde u\|_{L_2(0,l)}=\|u\|_{L_2(a,b)}$. Besides that,
\begin{multline*}
    (l_{ab}u,u)_{L_2(a,b)}=\|u'\|^2_{L_2(a,b)}+(qu,u)_{L_2(a,b)}
    \\=\|\tilde u'\|^2_{L_2(0,l)}+(q\tilde u,\tilde u)_{L_2(0,l)}=(l\tilde u,\tilde u)_{L_2(0,l)}.
\end{multline*}
Therefore
\begin{multline*}
    \min_{u\in\mathring H^1(a,b)}\frac{(l_{ab}u,u)_{L_2(a,b)}}{\|u\|^2_{L_2(a,b)}}=\min_{u\in\mathring H^1(a,b)}\frac{(l\tilde u,\tilde u)_{L_2(0,l)}}{\|\tilde u\|^2_{L_2(0,l)}}
    \\
    \geqslant
    \min_{u\in\mathring H^1(0,l)}\frac{(lu,u)_{L_2(0,l)}}{\|u\|^2_{L_2(0,l)}}=\lambda_1(L).
\end{multline*}
We obtained that $\lambda_1(L_{ab})\geqslant\lambda_1(L)>0$, which means that $0$ cannot be an eigenvalue of the operator $L_{ab}$, $u\notin{\rm Ker}\,L_{ab}$, a contradiction. Hence the function $u$ cannot have two zeros on the segment $[0,l]$, and the lemma is proved.
\end{proof}

Let us take an element $e\in{\rm Ker}\,L_0^*$ as a gauge element. The kernel of the operator $L_0^*$ consists of solutions of the equation $-u''+qu=0$, we take a solution which does not vanish at the point $\frac l2$ as $e$. The lemma just proved guarantees that $|e(x_{\omega})|^2+|e(l-x_{\omega})|^2\neq0$ for every $\omega\in\Omega_{L_0}$. The wave spectrum can be taken as the set $\Omega_{L_0}^e$. Indeed, let $u\in\mathcal U_{L_0}$ and $\omega\in\Omega_{L_0}$. Then
$$
    \frac{\|P_{\omega(t)}u\|^2}{\|P_{\omega(t)}e\|^2}=\frac{\int_{E_{x_{\omega}}(t)}|u(x)|^2dx}{\int_{E_{x_{\omega}}(t)}|e(x)|^2dx}\underset{t\to+0}
    \longrightarrow\frac{|u(x_{\omega})|^2+|u(l-x_{\omega})|^2}{|e(x_{\omega})|^2+|e(l-x_{\omega})|^2}.
$$
Thus Condition \ref{condition 5 gauge element} is satisfied. This allows to define on smooth waves the sesquilinear form
$$
    \langle u,v\rangle_{\omega}:=\frac{u(x_{\omega})\overline{v(x_{\omega})}+u(l-x_{\omega})\overline{v(l-x_{\omega})}}{|e(x_{\omega})|^2+|e(l-x_{\omega})|^2},\ u,v\in\mathcal U_{L_0}.
$$
Factorizing with respect to the equivalence relation
$$
    u\underset{\omega}\sim v\Leftrightarrow\langle u-v,u-v\rangle_{\omega}=0\Leftrightarrow
    \left\{
    \begin{array}{l}
    u(x_{\omega})=v(x_{\omega}),
    \\
    u(l-x_{\omega})=v(l-x_{\omega}),
    \end{array}
    \right.
$$
avoiding stalks, we come directly to spaces of values $\mathcal U_{L_0,\omega}^{\rm w}=\{[u](\omega),\ u\in\mathcal U_{L_0}\}$ of dimension two with the inner product
$$
    \langle[u](\omega),[v](\omega)\rangle_{\mathcal U_{L_0,\omega}^{\rm w}}=\frac{u(x_{\omega})\overline{v(x_{\omega})}+u(l-x_{\omega})\overline{v(l-x_{\omega})}}{|e(x_{\omega})|^2+|e(l-x_{\omega})|^2}.
$$
This definition does not depend on the choice of the equivalence class representatives $u$ and $v$. Denoting
$$
    \rho(x):=(|e(x)|^2+|e(l-x)|^2),
$$
we can write
\begin{multline*}
    (u,v)_{\mathcal H}=\int_0^lu(x)\overline{v(x)}dx=\int_0^{\frac l2}(u(x)\overline{v(x)}+u(l-x)\overline{v(l-x)})dx
    \\
    =\int_0^{\frac l2}\langle[u](\omega_x),[v](\omega_x)\rangle_{\mathcal U_{L_0,\omega_x}^{\rm w}}\rho(x)dx
    =\int_{\Omega_{L_0}}\langle[u](\omega),[v](\omega)\rangle_{\mathcal U_{L_0,\omega}^{\rm w}}d\mu(\omega),
\end{multline*}
where $\mu$ is the image of the measure $\rho(x)dx$ on the segment $[0,\frac l2]$ under the map $\beta$. Thus Condition \ref{condition 6 measure on spectrum} is satisfied. We obtain the space of the wave representation
$$
    \mathcal H_{L_0}^{\rm w}=\oplus\int_{\Omega_{L_0}}\mathcal U_{L_0,\omega}^{\rm w}d\mu(\omega).
$$

The operator $W^{\rm w}$ is the closure of the operator $W^{\rm w}_0:u\mapsto[u](\omega)$ defined on ${\rm Dom}\,W^{\rm w}_0=\mathcal U_{L_0}$. It is obviously isometric, but Condition \ref{condition 7 unitarity of transition operator} demands that it is unitary.

\begin{Lemma}\label{unitarity of W-w}
    The operator $W^{\rm w}$ is unitary.
\end{Lemma}

\begin{proof}
Let $y^{\rm w}\in\mathcal H^{\rm w}$ and $u\in\mathcal U_{L_0}$. For every $\omega\in\Omega_{L_0}$ the value $y^{\rm w}(\omega)$ is in $\mathcal U_{L_0,\omega}^{\rm w}$, the equivalence class of functions from $\mathcal U_{L_0}$, which have certain values at the points $x_{\omega}$ and $l-x_{\omega}$. Denote these values by $y(x_{\omega})$ and $y(l-x_{\omega})$. Then the function $y:[0,l]\to\mathbb C$ corresponds to the element $y^{\rm w}$, and
\begin{multline*}
    \langle[u](\omega),y^{\rm w}(\omega)\rangle_{\mathcal U_{L_0,\omega}^{\rm w}}
    =\langle[u](\omega),[v_y](\omega)\rangle_{\mathcal U_{L_0,\omega}^{\rm w}}
    \\
    =\frac{u(x_{\omega})\overline{v_y(x_{\omega})}+u(l-x_{\omega})\overline{v_y(l-x_{\omega})}}{\rho(x)}
    =\frac{u(x_{\omega})\overline{y(x_{\omega})}+u(l-x_{\omega})\overline{y(l-x_{\omega})}}{\rho(x)},
\end{multline*}
where $v_y$ is an element from $\mathcal U_{L_0}$ such that $[v_y](\omega)=y^{\rm w}(\omega)$. Therefore
\begin{multline*}
    \int_{\Omega_{L_0}}\langle[u](\omega),y^{\rm w}(\omega)\rangle_{\mathcal U_{L_0,\omega}^{\rm w}}d\mu(\omega)
    \\
    =\int_0^{\frac l2}\frac{u(x_{\omega})\overline{y(x_{\omega})}+u(l-x_{\omega})\overline{y(l-x_{\omega})}}{\rho(x)}\rho(x)dx
    =\int_0^lu(x_{\omega})\overline{y(x_{\omega})}dx
\end{multline*}
for every $u\in\mathcal U_{L_0}$, which implies that the integral on the right-hand side exists. This means that $y\in L_2(0,l)=\mathcal H$ and $y={W_0^{\rm w}}^*y^{\rm w}$. If $y=0$, then also $y^{\rm w}=0$, hence ${\rm Ker}\,{W_0^{\rm w}}^*=\{0\}$ and $\overline{{\rm Ran}\,W_0^{\rm w}}={\rm Ran}\,W^{\rm w}=\mathcal H$. Together with isometricity this means that $W^{\rm w}$ is a unitary operator, and the lemma is proved.
\end{proof}

\subsection{The coordinate representation}
Let $e_1$, $e_2$ be a basis in ${\rm Ker}\,L_0^*$. Solutions of the equation $-u''+qu=0$ are smooth functions, thus $e_1,e_2\in\mathcal U_{L_0}$.

\begin{Lemma}\label{Lemma base in value space}
    For every $\omega\in\Omega_{L_0}\setminus\{\omega_{\frac l2}\}$ the vectors $[e_1](\omega)$ and $[e_2](\omega)$ form a basis in $\mathcal U_{L_0,\omega}^{\rm w}$.
\end{Lemma}

\begin{proof}
Linear dependence of $[e_1](\omega)$ and $[e_2](\omega)$ would mean proportionality of the vectors
$
\begin{pmatrix}
   e_1(x_{\omega}) \\
   e_1(l-x_{\omega}) \\
\end{pmatrix}
$
and
$
\begin{pmatrix}
   e_2(x_{\omega}) \\
   e_2(l-x_{\omega}) \\
\end{pmatrix}
$
in $\mathbb C^2$, which would mean existence of a solution of $-u''+qu=0$ with zeros at the points $x_{\omega}$ and $l-x_{\omega}$. By Lemma \ref{Lemma not many zeros} this is impossible.
\end{proof}

It follows that Condition \ref{condition 8 basis} is satisfied with $\Omega_{L_0}^0=\Omega_{L_0}\setminus\{\omega_{\frac l2}\}$. Using the elements $e_1$, $e_2$, we define the coefficients
$$
    \hat u(x_{\omega})=
    \begin{pmatrix}
      \langle u,e_1\rangle_{\omega} \\
      \langle u,e_2\rangle_{\omega} \\
    \end{pmatrix}.
$$
in the spaces $\mathcal U_{L_0,\omega}^{\rm w}$. These coefficients are not the coordinates of the element $[u](\omega)$ in the decomposition over the basis $[e_1](\omega)$, $[e_2](\omega)$, such coordinates are given by the components of the vector $G^{-1}(\omega)\hat u(x_{\omega})$, where
$$
    G(\omega)=
    \begin{pmatrix}
      \langle e_1,e_1\rangle_{\omega} & \langle e_2,e_1\rangle_{\omega} \\
      \langle e_1,e_2\rangle_{\omega} & \langle e_2,e_2\rangle_{\omega} \\
    \end{pmatrix}
$$
is the Gram matrix. The coefficients $\hat u(x)\in\mathbb C^2$ are available to the ``outer observer'', the linear map $[u](\omega)\mapsto\hat u(x_{\omega})$ is bijective from $\mathcal U_{L_0,\omega}^{\rm w}$ to $\mathbb C^2$. In the space $\mathbb C^2$ we have to define an inner product corresponding to $\langle[u](\omega),[v](\omega)\rangle_{\mathcal U_{L_0,\omega}^{\rm w}}$ in $\mathcal U_{L_0,\omega}^{\rm w}$.

\begin{Lemma}\label{Lemma inner product}
    For every $u,v\in\mathcal U_{L_0,\omega}^{\rm w}$ and $\omega\in\Omega_{L_0}$
    $$
        \langle[u](\omega),[v](\omega)\rangle_{\mathcal U_{L_0,\omega}^{\rm w}}=(G^{-1}(\omega)\hat u(x_{\omega}),\hat v(x_{\omega}))_{\mathbb C^2}.
    $$
\end{Lemma}

\begin{proof}
By Lemma \ref{Lemma base in value space}, the Gram matrix $G(\omega_x)$ is non-degenerate for $x\in[0,\frac l2)$. Computation gives:
$$
    \langle[u](\omega),[v](\omega)\rangle_{\mathcal U_{L_0,\omega}^{\rm w}}
    =\frac{u(x_{\omega})\overline{v(x_{\omega})}+u(l-x_{\omega})\overline{v(l-x_{\omega})}}{\rho(x_{\omega})},
$$
$$
    \hat u(x_{\omega})=\frac1{\rho(x_{\omega})}
    \begin{pmatrix}
          u(x_{\omega})\overline{e_1(x_{\omega})}+u(l-x_{\omega})\overline{e_1(l-x_{\omega})}\\
          u(x_{\omega})\overline{e_2(x_{\omega})}+u(l-x_{\omega})\overline{e_2(l-x_{\omega})}\\
    \end{pmatrix}
    =T(x_{\omega})
    \begin{pmatrix}
         u(x_{\omega}) \\
         u(l-x_{\omega}) \\
    \end{pmatrix},
$$
where
\begin{equation}\label{T}
    T(x):=\frac1{\rho(x)}
    \begin{pmatrix}
          \overline{e_1(x)} & \overline{e_1(l-x)} \\
          \overline{e_2(x)} & \overline{e_2(l-x)} \\
    \end{pmatrix}.
\end{equation}
Furthermore,
\begin{multline*}
    (G^{-1}(\omega_x)\hat u(x),\hat v(x))_{\mathbb C^2}
    \\
    =
    \left(
    G^{-1}(\omega_x)T(x)
    \begin{pmatrix}
          u(x_{\omega}) \\
          u(l-x_{\omega}) \\
    \end{pmatrix}
    ,
    T(x)
    \begin{pmatrix}
          v(x_{\omega}) \\
          v(l-x_{\omega}) \\
    \end{pmatrix}
    \right)_{\mathbb C^2}
    \\
    =
    \left(
    T^*(x)G^{-1}(\omega_x)T(x)
    \begin{pmatrix}
          u(x_{\omega}) \\
          u(l-x_{\omega}) \\
    \end{pmatrix}
    ,
    \begin{pmatrix}
          v(x_{\omega}) \\
          v(l-x_{\omega}) \\
    \end{pmatrix}
    \right)_{\mathbb C^2}.
\end{multline*}
It is easy to see that $G(\omega_x)=\rho(x)T(x)T^*(x)$, so that $T^*(x)G^{-1}(\omega_x)T(x)=\frac{I}{\rho(x)}$ and
\begin{multline*}
    (G^{-1}(\omega)\hat u(x_{\omega}),\hat v(x_{\omega}))_{\mathbb C^2}
    \\
    =\frac1{\rho(x_{\omega})}
    \left(
    \begin{pmatrix}
          u(x_{\omega}) \\
          u(l-x_{\omega}) \\
    \end{pmatrix}
    ,
    \begin{pmatrix}
          v(x_{\omega}) \\
          v(l-x_{\omega}) \\
    \end{pmatrix}
    \right)_{\mathbb C^2}
    =\langle[u](\omega),[v](\omega)\rangle_{\mathcal U_{L_0,\omega}^{\rm w}},
\end{multline*}
and the lemma is proved.
\end{proof}

Consider the space of the coordinate representation
$$
    \mathcal H^{\rm c}:=L_2\left(\left(0,\frac l2\right),G^{-1}(\omega)\rho(x_{\omega})dx_{\omega}, \mathbb C^2\right).
$$
The operator $W_0^{\rm c}:u\mapsto\hat u$, from $\mathcal H$ to $\mathcal H^{\rm c}$, defined on ${\rm Dom}\,W_0^{\rm c}=\mathcal U_{L_0}$, after closure becomes an isometric operator $W^{\rm c}=\overline{W_0^{\rm c}}$ defined on the whole space $\mathcal H$.

\begin{Lemma}\label{Lemma W-c unitary}
The operator $W^{\rm c}$ is unitary and
\begin{equation}\label{W-c}
    (W^{\rm c}u)(x_{\omega})=T(x_{\omega})
    \begin{pmatrix}
          u(x_{\omega}) \\
          u(l-x_{\omega}) \\
    \end{pmatrix}.
\end{equation}
for every $u\in\mathcal H$.
\end{Lemma}

\begin{proof}
For every $u\in\mathcal U_{L_0}$ and $\hat y\in\mathcal H^{\rm c}$ using the equality
$$
    T^*(x_{\omega})G^{-1}(\omega)T(x_{\omega})=\frac{I}{\rho(x_{\omega})}
$$
we get:
\begin{multline*}
    (W_0^{\rm c}u,\hat y)_{\mathcal H^{\rm c}}=\int_0^{\frac l2}(G^{-1}(\omega_x)\hat u(x),\hat y(x))_{\mathbb C^2}\rho(x)dx
    \\
    =\int_0^{\frac l2}\left(G^{-1}(\omega_x)T(x)
    \begin{pmatrix}
          u(x) \\
          u(l-x) \\
    \end{pmatrix}
    ,\hat y(x)
    \right)_{\mathbb C^2}\rho(x)dx
    \\
    =
    \int_0^{\frac l2}\left(\frac{T^{-*}(x)}{\rho(x)}
    \begin{pmatrix}
          u(x) \\
          u(l-x) \\
    \end{pmatrix}
    ,\hat y(x)
    \right)_{\mathbb C^2}\rho(x)dx
    \\
    =\int_0^{\frac l2}\left(
    \begin{pmatrix}
          u(x) \\
          u(l-x) \\
    \end{pmatrix}
    ,T^{-1}(x)\hat y(x)
    \right)_{\mathbb C^2}dx=\int_0^lu(x)\overline{y(x)}dx,
\end{multline*}
where
$$
    y(x)=
    \left\{
    \begin{array}{ll}
        (T^{-1}(x)\hat y(x))_1,&x\in(0,\frac l2),
        \\
        (T^{-1}(l-x)\hat y(l-x))_2,&x\in(\frac l2,l).
    \end{array}
    \right.
$$
Note that $\hat y\in L_2((0,\frac l2),G^{-1}(\omega)\rho(x_{\omega})dx_{\omega},\mathbb C^2)$ means that
$$
    \int_0^{\frac l2}(G^{-1}(\omega)\hat y(x_{\omega}),\hat y(x_{\omega}))\rho(x_{\omega})dx_{\omega}=\int_0^{\frac l2}\|T^{-1}(x_{\omega})\hat y(x_{\omega})\|^2dx_{\omega}=\|y\|^2_{L_2(0,l)}.
$$
Therefore $y={W_0^{\rm c}}^*\hat y$. If $y=0$, then $T^{-1}\hat y=0$ and hence $\hat y =0$, so ${\rm Ker}\,{W_0^{\rm c}}^*=\{0\}$. This means that $\overline{{\rm Ran}\,W_0^{\rm c}}=\mathcal H^{\rm c}$ and the operator $\overline{W_0^{\rm c}}$ is unitary. Moreover, the operator which acts by the rule
$$
    u(x)\mapsto T(x)
    \begin{pmatrix}
          u(x) \\
          u(l-x) \\
    \end{pmatrix}
$$
from $\mathcal H$ to $\mathcal H^{\rm c}$ is isometric and coincides with $W_0^{\rm c}$ on $\mathcal U_{L_0}$. This implies that it is equal to $\overline{W_0^{\rm c}}$. Therefore \eqref{W-c} holds. The lemma is proved.
\end{proof}

Define the operator
$$
    L_0^{\rm c}=W^{\rm c}L_0{W^{\rm c}}^*
$$
in the space $\mathcal H^{\rm c}$. Owing to unitarity of $W^{\rm c}$,
\begin{multline*}
    {\rm Graph}\,{L_0^{\rm c}}^*=\overline{{\rm Graph\,}(W^{\rm c}L_0^*|_{\mathcal U_{L_0}}{W^{\rm c}}^*)}
    \\
    =\overline{\{(W^{\rm c}u^h(T),-W^{\rm c}u^{h_{tt}}(T)),h\in\mathcal M,T\geqslant0\}}
    \\
    =\overline{\{(\widehat{u^h(T)},-\widehat{u^{h_{tt}}(T)}),h\in\mathcal M,T\geqslant0\}}.
\end{multline*}
The ``outer observer'' can construct the graph of the operator ${L_0^{\rm c}}^*$ in this form using boundary control. This operator will be a differential operator of the second order, and one will be able to recover the original $L_0^*$ from it.

\begin{Theorem}
The operator ${L_0^{\rm c}}^*$ is defined on the domain
$$
    {\rm Dom}\,{L_0^{\rm c}}^*=\left\{\hat u(x)=T(x)
    \begin{pmatrix}
      u(x) \\
      u(l-x) \\
    \end{pmatrix},
    u\in H^2(0,l)
    \right\},
$$
where $T(x)$ is given by the formula \eqref{T} and acts by the rule
$$
    ({L_0^{\rm c}}^*\hat u)(x)=-\hat u''(x)+\hat P(x)\hat u'(x)+\hat Q(x)\hat u(x),
$$
where
\begin{equation}\label{P-hat}
    \hat P(x)=-2T(x){T^{-1}}'(x),
\end{equation}
\begin{equation}\label{Q-hat}
    \hat Q(x)=T(x)Q(x)T^{-1}(x)-T(x){T^{-1}}''(x),
\end{equation}
$$
    Q(x)=
    \begin{pmatrix}
          q(x) & 0 \\
          0 & q(l-x) \\
    \end{pmatrix}.
$$
Besides that,
$$
    {\rm Dom\,}L_0^{\rm c}=\{\hat u\in{\rm Dom\,}{L_0^{\rm c}}^*:\hat u(0)=\hat u'(0)=0\},
$$
$$
    {\rm Dom\,}(W^{\rm c}L{W^{\rm c}}^*)=\{\hat u\in{\rm Dom\,}{L_0^{\rm c}}^*:\hat u(0)=0\}.
$$
\end{Theorem}

\begin{proof}
For $u\in{\rm Dom\,}L_0^*$ we have:
$$
    \hat u(x)=(W^{\rm c}u)(x)=T(x)
    \begin{pmatrix}
          u(x) \\
          u(l-x) \\
    \end{pmatrix},
$$
\begin{multline*}
    ({L_0^{\rm c}}^*\hat u)(x)=\widehat{L_0^*u}(x)=T(x)
    \begin{pmatrix}
          -u''(x)+q(x)u(x) \\
          -u''(l-x)+q(l-x)u(l-x) \\
    \end{pmatrix}
    \\
    =T(x)
    \left(
    -
    \begin{pmatrix}
          u(x) \\
          u(l-x) \\
    \end{pmatrix}''
    +Q(x)
    \begin{pmatrix}
          u(x) \\
          u(l-x) \\
    \end{pmatrix}
    \right)
    \\
    =T(x)(-(T^{-1}(x)\hat u(x))''+Q(x)T^{-1}(x)\hat u(x))
    \\
    =-\hat u''(x)+\hat P(x)\hat u'(x)+\hat Q(x)\hat u(x).
\end{multline*}
Domains of the operators ${L_0^{\rm c}}^*$, $L_0^{\rm c}$, and $W^{\rm c}L{W^{\rm c}}^*$ can be easily found from the domains of the operators $L_0^*$, $L_0$, and $L$, respectively.
\end{proof}

\begin{Remark}
The domain of the operator ${L_0^{\rm c}}^*$ is contained in the linear set
$$
    \left\{\hat u\in H^2\left(\left[0,\frac l2\right],\mathbb C^2\right):\hat u\left(\frac l2\right)=\hat u_0
    \begin{pmatrix}
        \overline{e_1(\frac l2)} \\
        \overline{e_2(\frac l2)} \\
    \end{pmatrix},
    \hat u_0\in\mathbb C,\hat u'\left(\frac l2\right)=0
    \right\}.
$$
\end{Remark}

\begin{proof}
Since $T\in C^{\infty}[0,\frac l2]$,
$$
    T(x)
    \begin{pmatrix}
          u(x) \\
          u(l-x) \\
    \end{pmatrix}
    \in H^2([0,\frac l2],\mathbb C^2)
$$
holds for $u\in H^2(0,l)$. The vector-valued function
$v(x)=
\begin{pmatrix}
  u(x) \\
  u(l-x) \\
\end{pmatrix}
$,
besides belonging to $H^2([0,\frac l2],\mathbb C^2)$, satisfies two other conditions:
$v(\frac l2)=v_0
\begin{pmatrix}
  1 \\
  1 \\
\end{pmatrix}$
and
$v'(\frac l2)=v_1
\begin{pmatrix}
  1 \\
 -1 \\
\end{pmatrix}$
with some $v_0,v_1\in\mathbb C$. These conditions after multiplication by the matrix $T(x)$ turn into conditions
$\hat u(\frac l2)=\hat u_0
\begin{pmatrix}
  \overline{e_1(\frac l2)} \\
  \overline{e_2(\frac l2)} \\
\end{pmatrix}$
and $\hat u'(\frac l2)=0$, with $\hat u_0\in\mathbb C$. The first follows from substitution, for the second we used symmetry of the function $\rho(x)$ with respect to the point $\frac l2$.

The matrix $T(x)$ degenerates at the point $\frac l2$, hence $T^{-1}(x)\notin C^{\infty}[0,\frac l2]$ and only inclusion, not equality, of linear sets takes place.
\end{proof}

\subsection{The inverse problem}
The ``outer observer'' can recover the potential $q$ after construction of the wave model from the inverse data. But recovering is possible up to changing $q(x)$ to $q(l-x)$, which is natural: for these potentials the data will be the same. The wave model appears as a second order differential operator on the interval $(0,\frac l2)$, which acts on vector-valued functions with two components. Thus the coefficients $\hat P(x)$ and $\hat Q(x)$ are known. Note that the Gram matrix $G(\omega_x)$ and the density of the measure $\rho(x)=\lim\limits_{t\to+0}\frac{(P_{\omega_{x}(t)}e,e)}{2t}$ are determined in the ``wave'' terms and hence are available to the ``outer observer''.

To find the potential it is enough to know $\hat P$ and $\hat Q$. The equation $-2TT^{-1}=\hat P$ is equivalent to the equation ${T^{-1}}'=-\frac12T^{-1}\hat P$ on the function $T^{-1}$. Let $M(x)$ denote its fundamental (matrix) solution:
$$
    \begin{array}{l}
        M'(x)=-\frac12M(x)\hat P(x),
        \\
        M(0)=I.
    \end{array}
$$
Then $T^{-1}(x)=T_0^{-1}M(x)$ with some constant invertible matrix $T_0$ and $T(x)=M^{-1}(x)T_0$. Equation \eqref{Q-hat} reads
$$
    \hat Q=M^{-1}T_0QT_0^{-1}M-M^{-1}T_0(T_0M)'',
$$
which is equivalent to
$$
    M\hat QM^{-1}=T_0QT_0^{-1}-M''M^{-1},
$$
$$
    Q(x)=
        \begin{pmatrix}
          q(x) & 0 \\
          0 & q(l-x) \\
        \end{pmatrix}
    =T_0^{-1}(M(x)\hat Q(x)M^{-1}(x)+M''(x)M^{-1}(x))T_0.
$$
We see that the values of the potential $q$ at the points symmetric with respect to $\frac l2$ can be found as the eigenvalues of the matrix
$$
    M(x)\hat Q(x)M^{-1}(x)+M''(x)M^{-1}(x),
$$
and one can find this matrix from $\hat P$ and $\hat Q$. So we see that the potential can be recovered up to reflection from the middle of the interval.


\begin{thebibliography}{99}
\bibitem{Belishev-2013}
M. I. Belishev.
\newblock {A unitary invariant of a semi-bounded operator in reconstruction
of manifolds.}
\newblock {\em Journal of Operator Theory}, 69(2), 299-326, 2013.

\bibitem{Belishev-1997}
M. I. Belishev.
\newblock {Boundary control in reconstruction of manifolds and
metrics (the BC method).}
\newblock {\em Inverse Problems}, 13(5), 1--45, 1997.

\bibitem{Belishev-1988}
M. I. Belishev. On the Kac problem of the domain shape reconstruction via the Dirichlet problem spectrum. {\em Journal of Soviet Mathematics}, 55(3), 1663--1672, 1991.

\bibitem{Belishev-2007}
M. I. Belishev.
\newblock {Recent progress in the boundary control method.}
\newblock {\em Inverse Problems}, 23(5), 1--67, 2007.

\bibitem{Belishev-Demchenko-2012}
M. I. Belishev, M. N. Demchenko. Dynamical system with boundary control associated with a symmetric semibounded operator. {\em Journal of Mathematical Sciences}, 194(1), 8--20, 2013. DOI:10.1007/s10958-013-1501-8.

\bibitem{Belishev-Demchenko-2014}
M. I. Belishev, M. N. Demchenko.
\newblock{Elements of noncommutative geometry in inverse problems on manifolds.}
\newblock{\em Journal of Geometry and Physics}, 78, 29--47, 2014.

\bibitem{Belishev-Simonov-2017}
M. I. Belishev, S. A. Simonov. Wave model of the Sturm-Liouville operator on the half-line. {\em St. Petersburg Math. J.}, 29(2), 227--248, 2018.

\bibitem{Birkhoff-1984}
G. Birkhoff. Lattice Theory. {\em Providence, Rhode Island}, 1967.

\bibitem{Birman-Solomyak-1980}
M. S. Birman, M. Z. Solomyak. Spectral Theory of Self-Adjoint Operators in Hilbert Space. {\em D.Reidel Publishing Comp.}, 1987.

\bibitem{Derkach-Malamud-1995}
V. A. Derkach, M. M. Malamud. The extension theory of Hermitian
operators and the moment problem. {\em Journal of Mathematical
Sciences}, 73(2), 141--242, 1995.

\bibitem{Kelley-1981}
J. L. Kelley. General Topology. {\em D.Van Nostrand Company, Inc. Princeton, New Jersey, Toronto, London, New York}, 1957.

\bibitem{Kim-2006}
J. M. Kim.
Compactness in $\mathcal B(X)$.
{\em J. Math. Anal. Appl.}, 320, 619--631, 2006.

\bibitem{Kochubei-1975}
A. N. Kochubei. Extensions of symmetric operators and symmetric binary relations. {\em Math. Notes}, 17(1), 25--28, 1975.

\bibitem{Kolmogorov-Fomin-1989}
A. N. Kolmogorov, S. V. Fomin. Elements of the theory of functions and functional analysis. Vol. 1, Metric and normed spaces. {\em Graylock Press}, 1957.

\bibitem{Naimark-1969}
M. A. Naimark. Linear Differential Operators. {\em WN Publishing, Gronnongen, The Netherlands}, 1970.

\bibitem{Ryzhov-2007}
V. Ryzhov.
\newblock {A general boundary value problem and its Weyl function.}
\newblock {\em Opuscula Math.}, 27(2), 305--331, 2007.

\bibitem{Simonov-2017}
S. A. Simonov.
\newblock {Wave model of the regular Sturm--Liouville operator}.
\newblock{\em Proceedings of 2017 Days on Diffraction}, 300--303, 2017. arXiv: 1801.02011.

\bibitem{Straus-1998}
A. V. Strauss. Functional models and generalized spectral functions of symmetric operators. {\em St. Petersbg. Math. J.}, 10(5), 733--784, 1999.

\bibitem{Vishik-1952}
M. I. Vishik. On general boundary problems for elliptic differential equations. {\em Transl., Ser. 2, Am. Math. Soc.}, 24, 107--172, 1963.

\end{thebibliography}
\end{document}